\documentclass[prx,aps,twocolumn,floatfix,amsmath,notitlepage,longbibliography]{revtex4-1}

\usepackage{amssymb}
\usepackage{graphicx}
\usepackage{graphics}
\usepackage{amsmath}
\usepackage{amsthm}
\usepackage{color}
\usepackage{dsfont}
\usepackage{mathrsfs}

\usepackage{mathtools}
\usepackage[unicode=true,pdfusetitle,bookmarks=true,bookmarksnumbered=false,bookmarksopen=false,breaklinks=false,pdfborder={0 0 0},pdfborderstyle={},backref=false,colorlinks=true]
{hyperref}

\def\>{\rangle}
\def\<{\langle}

\def\id{\mathsf{id}}

\def\mE{\mathcal{E}}

\def\mF{\mathcal{F}}
\def\mN{\mathcal{N}}
\def\mC{\mathcal{C}}
\def\mL{\mathcal{L}}
\def\mP{\mathcal{P}}

\def\mD{\mathcal{D}}

\def\mV{\mathcal{V}}

\newcommand{\supp}{\operatorname{supp}}

\renewcommand{\qedsymbol}{\nobreak \ifvmode \relax \else
	\ifdim \lastskip<1.5em \hskip-\lastskip \hskip1.5em plus0em
	minus0.5em \fi \nobreak \vrule height0.75em width0.5em
	depth0.25em\fi}

\renewcommand{\geq}{\geqslant}
\renewcommand{\leq}{\leqslant}

\def\md{\mathfrak{D}}

\def\mr{\mathfrak{R}}

\def\D{\mathbf{D}}

\def\C{\mathbf{C}}

\newcommand{\super}{{\rm SC}}

\newcommand{\pr}{{\rm Pr}}

\def\ml{\mathfrak{L}}
\def\mt{\mathfrak{T}}

\def\sD{\mathbb{D}}
\def\xD{\mathscr{D}}

\newtheorem{theorem}{Theorem}
\newtheorem*{theorem*}{Theorem}
\newtheorem{corollary}{Corollary}
\newtheorem{lemma}{Lemma}
\newtheorem*{lemma*}{Lemma}

\theoremstyle{definition}
\newtheorem{definition}{Definition}
\newtheorem*{definition*}{Definition}

\theoremstyle{remark}
\newtheorem*{remark}{Remark}

\theoremstyle{definition}

\newcommand{\bea}{\begin{eqnarray}}
\newcommand{\eea}{\end{eqnarray}}
\newcommand{\be}{\begin{equation}}
\newcommand{\ee}{\end{equation}}
\newcommand{\ba}{\begin{equation}\begin{aligned}}
\newcommand{\ea}{\end{aligned}\end{equation}}

\newcommand{\bsp}{\begin{split}}
\newcommand{\esp}{\end{split}}

\def\be{\begin{equation}}
\def\ee{\end{equation}}

\newcommand{\cptp}{{\rm CPTP}}
\newcommand{\cp}{{\rm CP}}
\newcommand{\reg}{{\rm reg}}

\newcommand{\mR}{\mathcal{R}}
\newcommand{\mM}{\mathcal{M}}

\newcommand{\lr}{\rangle\langle}
\newcommand{\la}{\langle}
\newcommand{\ra}{\rangle}
\newcommand{\tr}{{\rm Tr}}

\newcommand{\ve}[1]{{\left\vert\kern-0.25ex\left\vert\kern-0.25ex\left\vert #1 
    \right\vert\kern-0.25ex\right\vert\kern-0.25ex\right\vert}}

\newcommand{\mbb}[1]{\mathbb{#1}}

\newcommand{\eqdef}{\coloneqq}

\newcommand{\bb}{\begin{bmatrix}}
\newcommand{\eb}{\end{bmatrix}}

\def\r{\mathbf{r}}
\def\s{\mathbf{s}}
\def\p{\mathbf{p}}
\def\q{\mathbf{q}}

\def\t{\mathbf{t}}
\def\u{\mathbf{u}}
\def\v{\mathbf{v}}
\def\w{\mathbf{w}}
\def\0{\mathbf{0}}

\def\m{\mathbf{m}}
\def\n{\mathbf{n}}

\def\uD{\underbar{D}}
\def\bD{\overline{D}}

\def\bbd{\overline{\mathbf{D}}}
\def\ubd{\underline{\mathbf{D}}}

\def\mX{\mathcal{X}}

\usepackage[most]{tcolorbox}
\newtcolorbox{myt}[2][]{%
  attach boxed title to top center
               = {yshift=-4pt},
  colback      = blue!5!white,
  colframe     = blue!75!black,
  halign       = flush left,
  fonttitle    = \bfseries\sffamily,
  colbacktitle = blue!65!black,
  title        = #2,#1,
  enhanced,
}
\newtcolorbox{myd}[2][]{%
  attach boxed title to top center
               = {yshift=-4pt},
  colback      = violet!5!white,
  colframe     = violet!75!black,
  halign       = flush left,
  fonttitle    = \bfseries\sffamily,
  colbacktitle = violet!65!black,
  title        = #2,#1,
  enhanced,
}
\newtcolorbox{mye}[2][]{%
  attach boxed title to top center
               = {yshift=-4pt},
  colback      = purple!5!white,
  colframe     = purple!75!black,
  halign       = flush left,
  fonttitle    = \bfseries\sffamily,
  colbacktitle = purple!65!black,
  title        = #2,#1,
  enhanced,
}

\newtcolorbox{myg}[2][]{%
  attach boxed title to top center
               = {yshift=-4pt},
  colback      = green!5!white,
  colframe     = green!75!black,
  halign       = flush left,
  fonttitle    = \bfseries\sffamily,
  colbacktitle = green!65!black,
  title        = #2,#1,
  enhanced,
}

\newcommand{\eps}{\varepsilon}

\begin{document}
	
	
	\title{Uniqueness and Optimality of Dynamical Extensions of Divergences}

\author{Gilad Gour}\email{gour@ucalgary.ca}
\affiliation{
Department of Mathematics and Statistics, Institute for Quantum Science and Technology,
University of Calgary, AB, Canada T2N 1N4}

	\date{\today}
	
	\begin{abstract}
	We introduce an axiomatic approach for channel divergences and channel relative entropies that is based on three information-theoretic axioms of monotonicity under superchannels (i.e. generalized data processing inequality), additivity under tensor products, and normalization, similar to the approach given for the state domain in~\cite{GT2020a,GT2020b}. We show that these axioms  are sufficient to give enough structure also in the channel domain, leading to numerous properties that are applicable to all channel divergences. These include faithfulness, continuity, a type of triangle inequality, and boundedness between the min and max channel relative entropies. In addition, we prove a uniqueness theorem showing that the Kullback-Leibler divergence has only one extension to classical channels. For quantum channels, with the exception of the max relative entropy, this uniqueness does not hold. Instead we prove the optimality of the amortized channel extension of the Umegaki relative entropy, by showing that it provides a lower bound on all channel relative entropies that reduce to the Kullback-Leibler divergence on classical states. We also introduce the maximal channel extension of a given classical state divergence and study its properties.  	
	\end{abstract}

	\maketitle
{
  \hypersetup{linkcolor=black}
 \tableofcontents
}

	\section{Introduction}
Distinguishability between probability distributions, quantum states, and quantum channels, lies at the heart of several theories in science including information theory, quantum computing, and statistics. Unlike arbitrary vectors, objects like probability vectors, quantum states and quantum channels contain information about physical systems and therefore their distinguishability is typically quantified with functions that are sensitive to this information. For example,
consider an information source (see e.g.~\cite{Cover2006}) such as a dice, that `emits' alphabets in $\mathcal{X}$ according to either probability distribution $\p\eqdef\{p_x\}_{x\in\mX}$ or probability distribution $\q\eqdef\{q_x\}_{x\in\mX}$. An agent can estimate which of the two distributions corresponds to the source by using the source many times (in the example of a dice, this amounts to rolling it many times). The intuition is that if $\p$ and $\q$ are very distinguishable it would be easier to determine which one of them corresponds to the information source. 

One of the key observations in any distinguishability task as above, is that by sending the information source through a communication channel, the agent cannot increase his or her ability to distinguish between the two distributions $\p$ and $\q$. This means that if $E$ is a column stochastic matrix that corresponds to the communication channel, then the resulting distributions $E\p$ and $E\q$ are { less distinguishable than $\p$ and $\q$.} In other words, any measure that quantifies the distinguishability between two probability vectors $\p$ and $\q$ must decrease under any stochastic evolution process that takes the pair $(\p,\q)$ to $(E\p,E\q)$. { Hence, many metrics in $\mbb{R}^n$, such as the Euclidean distance, cannot be used to quantify this notion of distinguishability between two probability vectors (see e.g.~\cite{PWP+2006}).}

Functions defined on pairs of  probabilities that behave monotonically under mapping of the form $(\p,\q)\to(E\p,E\q)$ are called \emph{divergences}, and their monotonicity property is often referred to as \emph{data processing inequality} (DPI).  { Unlike the trace distance~\cite{Wilde2013}, some divergences are additive under tensor products. }In~\cite{GT2020a,GT2020b}, relative entropies were defined axiomatically as divergences that satisfy this additivity property. It was shown that the combination of DPI with additivity provides enough structure to deduce a variety of key properties of relative entropies.

Many specific divergences and relative entropies have been extended to the quantum domain, including the trace distance, the Kullback-Leibler divergence, the Renyi divergences, $f$-divergences, and the min and max relative entropies~\cite{Umegaki1962,Petz-1986a,WWY2014,MDS+2013,Tomamichel2015,Matsumoto2018,Hiai2017,Datta-2009a}.
More recently, several extensions of quantum divergences to channel divergences have also been studied~\cite{Cooney-2016a,Felix2018,Berta2018,Gour2019,Fang2019,Katariya2020} with a variety of applications in resource theories of quantum channels and beyond~\cite{Kaur-2018a,Gour2018b,LW2019,LY2020,GW2019,Wang2019,Wang2019b,GS2020,Bauml2020,Fang2019,Fang2020,Katariya2020}. In~\cite{GT2020b}, an axiomatic approach to quantum divergences have been proposed. 
This approach is based on three axioms of DPI, additivity, and normalization that were first put forth 
in~\cite{GT2020a} for classical divergences.	

In this paper we extend further this axiomatic approach to the channel domain. We define dynamical (i.e. channel) divergences as functions that satisfies the DPI with a superchannel, and channel relative entropies as channel divergences that satisfy in addition additivity and appropriate normalization. We show that these axioms provides enough structure that give rise to a variety 
of properties that hold for all dynamical divergences and dynamical relative entropies. These include properties that carries over from the state domain like, faithfulness, continuity, optimality, and triangle inequality that holds for all channel relative entropies. We also show that all channel relative entropies must be no greater than the max channel relative entropy, and no smaller than the min channel relative entropy. We then discuss a variety of extensions of classical (state) divergences to quantum-channel divergences. Remarkably, for classical channels, we find that there exists only one classical-channel relative entropy that reduces to the Kullback-Leibler (KL) divergence. In other words, the extension of the KL relative entropy to classical channels is unique. 
In the quantum domain, we show that the regularized channel (Umegaki) relative entropy is the smallest extension among all dynamical relative entropies that reduces to the KL-divergence on classical states. In Fig.~\ref{channel} we summarized all the optimal extensions of the Shannon entropy that we  discuss in this paper. Similar figures can be made also for the Re\'nyi entropies.

\begin{figure}[h]\centering
    \includegraphics[width=0.45\textwidth]{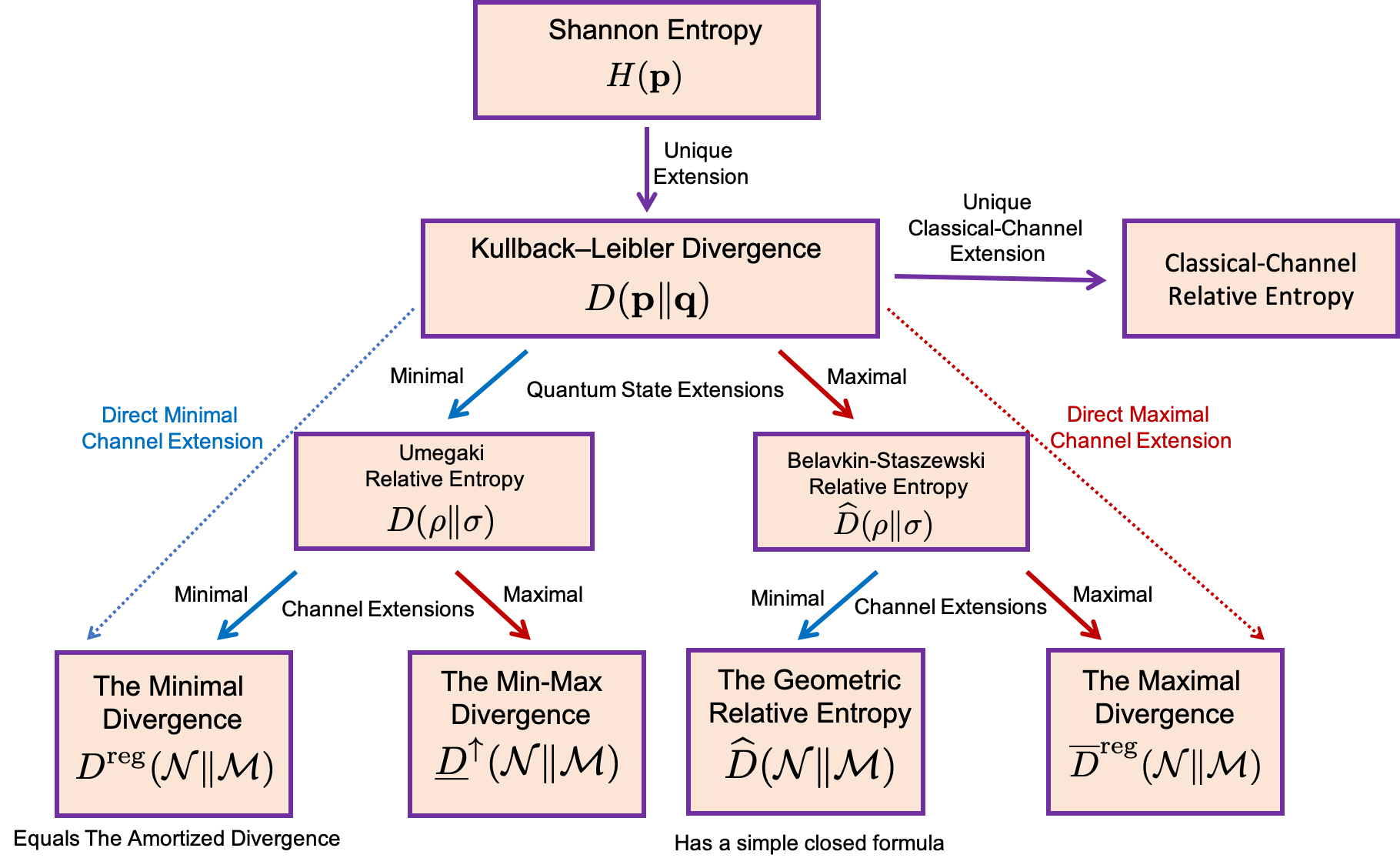}
  \caption{\linespread{1}\selectfont{\small A summary of all optimal extensions of the Shannon entropy. Red arrows represent maximal extensions and blue arrows represent minimal extensions. The Kullback-Leibler divergence is the only classical relative entropy that reduces to the log of the dimension minus the Shannon entropy when the second argument is a uniform distribution (in~\cite{GT2020a} a one-to-one correspondence between classical entropies and classical relative entropies was proven).
  Regularization is assumed in all extensions. The channel relative entropy $D^{\reg}(\mN\|\mM)$ equals both the amortized divergence and the regularized minimal channel extension $\underline{D}^{\reg}(\mN\|\mM)$. It is the smallest of all channel relative entropies that reduces to the Kullback-Leibler divergence on classical states, while the maximal divergence $\overline{D}^{\reg}(\mN\|\mM)$ is the largest one. The blue and red dashed arrows indicate that the minimal and maximal channel-extensions can be obtained directly from the Kullback-Leibler divergence using the extension techniques introduced in this paper.
   }}
  \label{channel}
\end{figure}

The paper is structured as follows. In Sec.II we discuss the preliminaries material and the notations that will be used throughout the paper. { We then discuss the main results in sections III and IV, and conclude with a summary and applications in section V.} The details of all proofs can be found in the supplemental material (SM).

\section{Preliminaries}

\subsection{Notations}

We denote both physical systems and their corresponding Hilbert spaces by $A,B,C$ etc. The letters $X,Y,Z$ are reserved for classical systems.
We will only consider finite dimensional systems and denote their dimensions by $|A|$, $|B|$, etc.
The algebra of all $|A|\times|A|$ complex matrices is denoted by $\ml(A)$. Similarly, $\ml(X)$ denotes the algebra of all $|X|\times |X|$ diagonal matrices. The set of density matrices (i.e. quantum states) in $\ml(A)$ is denoted by $\md(A)$. The elements of $\md(A)$ will be denoted by lower case Greek letters such as $\rho,\sigma,\omega$, etc, whereas for a classical system $X$, the elements of classical density matrices in $\md(X)$ are denoted by $\p,\q,\r$, etc, where depending on the context, $\p,\q,\r$ will be viewed as either diagonal density matrices in a fixed (classical) basis or as probability vectors. Moreover, for classical systems we will sometimes use the notation $\md(n)\equiv\md(X)$ if $n=|X|$.
The support of a density matrix $\rho$ will be denoted by $\supp(\rho)$.

The set of all linear transformations from 
$\ml(A)$ to $\ml(B)$ is denoted by $\ml(A\to B)$. The elements of $\ml(A\to B)$ will be denoted by calligraphic letters such as $\mN,\mM,\mE,\mF$, etc. The set of completely positive maps in $\ml(A\to B)$ is denoted by $\cp(A\to B)$, and the set of all completely positive trace-preserving maps (i.e. quantum channels) by $\cptp(A\to B)$. We denote by $1$ the trivial physical system and identify $\md(A)=\cptp(1\to A)$. In particular, the set containing only the number one is identified with $\{1\}=\md(1)=\cptp(1\to 1)$.
Note also that with these identifications, the trace is the only element of $\cptp(A\to 1)=\{\tr_A\}$. The Choi matrix of a linear map $\mN\in\ml(A\to B)$ will be denoted by $J_\mN^{AB}\eqdef\sum_{j,k}|j\lr k|\otimes\mN(|j\lr k|)$. For two linear maps $\mN,\mM\in\ml(A\to B)$ we write $\mN\geq\mM$ if $\mN-\mM\in\cp(A\to B)$.

We also consider linear maps from $\ml(A\to B)$ to $\ml(A'\to B')$. We call such linear maps supermaps, and denote by $\mbb{L}(AB\to A'B')$ the set of all supermaps from $\ml(A\to B)$ to $\ml(A'\to B')$. We will also make the identification $\mbb{L}(1A\to 1B)=\ml(A\to B)$, where again, $1$ denotes the trivial system. The elements of $\mbb{L}(AB\to A'B')$ will be denoted by capital Greek letters such as $\Theta$ and $\Upsilon$. A supermap $\Theta\in\mbb{L}(AB\to A'B')$ is called a superchannel if there exists { a reference system $R$,} a pre-processing map $\mE\in\cptp(A'\to AR)$, and a post-processing map $\mF\in\cptp(BR\to B')$ such that for any $\mN\in\ml(A\to B)$
\be
\Theta\left[\mN^{A\to B}\right]=\mF^{BR\to B'}\circ\mN^{A\to B}\circ\mE^{A'\to AR}\;.
\ee 
In~\cite{Pavia1} it has been shown that a supermap is a superchannel if and only if it maps quantum channels to quantum channels in a complete sense (i.e. even when tensored with the identity supermap; see~\cite{Gour2019,GS2020,BGS+2020} for more details). The set of all superchannels in $\mbb{L}(AB\to A'B')$ will be denoted by $\super(AB\to A'B')$. We will also use the notation
\be
\super(A\to A'B')\eqdef\super(1A\to A'B')
\ee
to denote superchannels that map quantum states in $\md(A)=\cptp(1\to A)$ to quantum channels in $\cptp(A'\to B')$. Similarly, $\super(AB\to A')$ denotes the set of all superchannels that map channels in $\cptp(A\to B)$ to states in $\md(A')$.

\subsection{Divergences} 

We follow the definition given in~\cite{GT2020a,GT2020b} of classical and quantum divergences, and relative entropies. 

\begin{definition*}\label{defd}\cite{GT2020a,GT2020b}
Let $\D:\bigcup_{A}\Big\{\md(A)\times\md(A)\Big\}\to\mbb{R}\cup\{\infty\}$ be a function acting on pairs of quantum states in all finite dimensions.
\begin{enumerate}
\item The function $\D$ is called a \emph{divergence} if it satisfies the data processing inequality (DPI)
\be
\D\big(\mE(\rho)\big\|\mE(\sigma)\big)\leq \D(\rho\|\sigma)\;,
\ee
for all $\mE\in\cptp(A\to B)$ and all $\rho,\sigma\in\md(A)$.
\item A divergence $\D$ is called a \emph{relative entropy} if in addition it satisfies:
\begin{enumerate}
\item Normalization. {
\be
\D\left(\begin{bmatrix} 1& 0\\ 0 & 0\end{bmatrix}\;\Big\|\; \begin{bmatrix} 1/2 & 0\\ 0 &1/2\end{bmatrix}\right)=1\;.
\ee}
\item Additivity. For any $\rho_1,\rho_2\in\md(A)$ and any $\sigma_1, \sigma_2\in\md(B)$
\be
\D\left(\rho_1\otimes\sigma_1\big\|\rho_2\otimes\sigma_2\right)= \D(\rho_1\|\sigma_1)+\D(\rho_2\|\sigma_2)\;.
\ee
\end{enumerate}
\end{enumerate}
\end{definition*}

{ We use the bold $\D$ notation to denote both a general \emph{quantum} divergence and a general quantum relative entropy depends on the context. Similarly, we use the notation $\xD$ to denote a general {classical} divergence or a general classical relative entropy.} For specific divergences, we use the standard notations to denote them. For example, the classical R\'enyi divergence of order $\alpha\in[0,\infty]$ is denoted by
\be
D_\alpha(\p\|\q)\eqdef\frac1{\alpha-1}\log\sum_{x=1}^{|X|}p_x^\alpha q_x^{1-\alpha}\quad\quad\forall\;\p,\q\in\md(X)\;.
\ee 
The min and max quantum relative entropies~\cite{Datta2009} are denoted for any $\rho,\sigma\in\md(A)$ by {
\begin{align}
D_{\max}(\rho\|\sigma)\eqdef\begin{cases}\log\min\big\{t\geq 0\;:\;t\sigma\geq\rho\big\}&\text{if }\rho\ll\sigma\\
\infty &\text{otherwise}
\end{cases}
\end{align}
where $\rho\ll\sigma$ stands for $\supp(\rho)\subseteq\supp(\sigma)$, and
\begin{align}
D_{\min}(\rho\|\sigma)\eqdef\begin{cases}-\log\tr\big[\sigma\Pi_{\rho}\big] &\text{if }\tr[\rho\sigma]\neq 0\\
\infty &\text{otherwise}
\end{cases}
\end{align}}where $\Pi_\rho$ denotes the projection to the support of $\rho$, and $\rho\ll\sigma$ denotes $\supp(\rho)\subseteq\supp(\sigma)$. In~\cite{GT2020b} it has been shown that any quantum relative entropy $\D$ satisfies
\be
D_{\min}(\rho\|\sigma)\leq\D(\rho\|\sigma)\leq D_{\max}(\rho\|\sigma)\quad\quad\forall\;\rho,\sigma\in\md(A)\;.
\ee
{ The Hypothesis testing divergence, can be viewed as the smoothed version of $D_{\min}$. It is given for all $\rho,\sigma\in\md(A)$ by
\begin{align}\label{ht}
&D_{\min}^\epsilon(\rho\|\sigma)\eqdef-\log\min\\
&\Big\{\tr[\sigma E]\;:\;0\leq E\leq I^A\;\;,\;\;\tr[\rho E]\geq 1-\epsilon\;\;,\;\;E\in\ml(A)\Big\}.\nonumber
\end{align}
It can be easily verified that $D_{\min}^{\eps=0}(\rho\|\sigma)=D_{\min}(\rho\|\sigma)$.}

Another important example of a quantum relative entropy with several operational interpretations is the Umegaki relative entropy, given by
\be
D(\rho\|\sigma)\eqdef\tr\left[\rho\log\rho\right]-\tr[\rho\log\sigma]\;.
\ee
This divergence plays a key role in quantum statistics~\cite{Petz2008}, quantum Shannon theory~\cite{Wilde2013}, and quantum resource theories~\cite{CG2019}. In~\cite{Matsumoto2018b} (see also~\cite{GT2020b,WGE2017}) it has been shown to be the only relative entropy that is asymptotically continuous.The Umegaki relative entropy is also known to be the regularized minimal quantum extension of the classical KL divergence; that is, for all $\rho,\sigma\in\md(A)$
\be\label{wer}
D(\rho\|\sigma)=\uD^{\reg}(\rho\|\sigma)\eqdef\lim_{n\to\infty}\frac1n\uD(\rho^{\otimes n}\big\|\sigma^{\otimes n})\;,
\ee
with  { 
\be
\uD(\rho\|\sigma)\eqdef\sup_{\mE\in\cptp(A\to X)}D_{KL}\left(\mE(\rho)\big\|\mE(\sigma)\right)\;,
\ee
where $D_{KL}$ is the classical KL-divergence, }and the supremum is over all classical systems $X$ and all  positive operator-valued measures (POVMs) $\mE\in\cptp(A\to X)$.  The quantity $\uD$ is usually known as the \emph{measured} relative entropy~\cite{Tomamichel2015}. The equality in~\eqref{wer} means that any quantum divergence that reduces to the classical KL relative entropy must be no smaller than the Umegaki relative entropy (see~\cite{GT2020b} for more details on optimal extensions). In~\cite{Tomamichel2015} it was shown that~\eqref{wer} also hold if $D$ is replaced with the sandwiched or minimal quantum R\'enyi divergence of order $\alpha\in[1/2,\infty]$~\cite{WWY2014,MDS+2013}.

Recently, the extensions of quantum divergences to channel divergences have been studied intensively~\cite{Cooney-2016a,Felix2018,Berta2018,Gour2019,Fang2019,Katariya2020}. Examples include the channel extension of the Umegaki relative entropy given for all $\mN,\mM\in\cptp(A\to B)$ by 
\ba
D(\mM\|\mN)\eqdef\sup_{\rho\in\md(RA)}D\left(\mM^{A\to B}(\rho^{RA})\big\|\mN^{A\to B}(\rho^{RA}\right)\;,
\ea
where the supremum is over all systems $R$ and all density matrices in $\md(RA)$. It can be shown~\cite{Cooney-2016a} that the supremum above can be replaced with a maximum over all pure states in $\md(RA)$ with $|R|=|A|$.
The min and max relative entropies have also been extended to quantum channels. They are given for all $\mN,\mM\in\cptp(A\to B)$ by
\be
D_{\max}(\mM\|\mN)\eqdef\log\min\big\{t\geq 0\;:\;t\mN\geq\mM\big\}\label{max1}
\ee
(recall that $t\mN\geq\mM$ means that $t\mN-\mM$ is a CP map)
if $\supp(J_{\mN})\subseteq\supp(J_\mM)$ and $D_{\max}(\mM\|\mN)=\infty$ otherwise, and
\be\label{min2}
D_{\min}(\mM\|\mN)\eqdef\max_{\psi\in\md(RA)}\Big\{-\log\tr\big[\mN(\psi)\Pi_{\mM(\psi)}\big]\Big\}
\ee
if for all $\psi\in\md(RA)$, $\tr[\mN(\psi)\mM(\psi)]\neq 0$, and $D_{\min}(\mM\|\mN)=\infty$ otherwise.
We assumed above that $|R|=|A|$. The divergence $D_{\min}(\mM\|\mN)$ can be viewed as the $\epsilon=0$ case of the Hypothesis testing channel divergence given by
\ba
&D_{\min}^{\epsilon}(\mM\|\mN)\\
&\eqdef\max_{\psi\in\md(RA)}D_{\min}^{\epsilon}\left(\mM^{A\to B}(\psi^{RA})\big\|\mN^{A\to B}(\psi^{RA})\right)\label{min2}
\ea
{ where on the LHS $D_{\min}^\eps$ corresponds to a channel divergence, whereas on the RHS  $D_{\min}^\eps$ corresponds to a state divergence. We can assume above that $|R|=|A|$, the maximum is over all pure states in $\md(RA)$, and for any $\rho,\sigma\in\md(A)$.}
We will discuss more examples of channel divergences later on.

\subsection{Relative Majorization}

We say that a pair of vectors $\p,\q\in\md(X)$ is \emph{relatively majorized} by another pair of vectors $\p',\q'\in\md(Y)$, and write
\be
(\p,\q)\succ(\p',\q')
\ee 
if there exists a stochastic evolution matrix $\mE\in\cptp(X\to Y)$ such that $(\p',\q')=(\mE(\p),\mE(\q))$. Therefore, the relative R\'enyi entropies behave monotonically under relative majorization; i.e. if $(\p,\q)\succ(\p',\q')$ then $D_\alpha(\p\|\q)\geq D_\alpha\big(\p'\big\|\q'\big)$ for all $\alpha\in[0,\infty]$. If we have $(\p,\q)\succ(\p',\q')\succ(\p,\q)$ then we will write 
\be
(\p',\q')\sim(\p,\q)\;.
\ee 

Relative majorization is a partial order that can be characterized with testing regions. The testing region of a pair of probability vectors $\p,\q\in\mD(n)$ is a region in $\mbb{R}^2$ defined as
\ba
&\mt(\p,\q)\eqdef\\
&\Big\{(\p\cdot\t,\q\cdot\t)\in\mbb{R}^2\;:\;0\leq \t\leq (1,...,1)^T\;,\;\t\in\mbb{R}^n\Big\}
\ea
where the inequalities are entry-wise. This region is bounded by two curves known as the lower and upper Lorenz curves. 

{ The upper Lorenz curve can be obtained from the lower Lorenz curve by a rotation of a 180 degrees around the point $(\frac12,\frac12)\in\mbb{R}^2$.} Therefore, the lower (or upper) Lorenz curve determines uniquely the testing region. The lower Lorenz curve of a pair of probability vectors $\p,\q\in\mbb{R}^n$, denoted here by $\mL(\p,\q)$, has $n+1$ vertices that can be computed as follows.  First, observe that the testing region is invariant under the transformation $(\p,\q)\to (\Pi\p,\Pi\q)$ where $\Pi$ is any permutation matrix. Therefore, w.l.o.g. we can assume that the components of $\p$ and $\q$ are arranged such that
\be
\frac{p_1}{q_1}\geq\frac{p_2}{q_2}\geq\cdots\geq\frac{p_n}{q_n}\;.
\ee
The  $n+1$ vertices of $\mL(\p,\q)$ are the $(0,0)$ vertex and the $n$ vertices $\{(a_\ell,b_\ell)\}_{\ell=1}^{n}\subset \mL(\p,\q)$, where 
\be
a_\ell=\sum_{x=1}^{\ell}p_x\quad\text{and}\quad b_\ell=\sum_{x=1}^{\ell}q_x\;.
\ee

The relevance of testing regions to our study here is the following theorem that goes back to Blackwell~\cite{blackwell1953}, and that since then have been rediscovered under different names including $d$-majorization~\cite{Veinott1971}, matrix majorization~\cite{Dahl-1999}, and thermo-majorization~\cite{Horodecki-2013b} (see also the book on majorization by 
Marshall and Olkin~\cite{Marshall-2011a}).
\begin{theorem*}\cite{blackwell1953}
Let $\p,\q\in\mP(n)$ and $\p',\q'\in\mP(m)$ be two pairs of probability vectors in dimensions $n$ and $m$, respectively. Then, 
\be
(\p,\q)\succ(\p',\q')\quad\iff\quad\mt(\p,\q)\supseteq\mt(\p',\q')\;.
\ee
\end{theorem*}
The theorem above provides a geometric characterization to relative majorization; that is,  $(\p,\q)\succ(\p',\q')$ if and only if the lower Lorenz curve of $(\p',\q')$ is nowhere below the lower Lorenz curve of $(\p,\q)$.
Relative majorization has another remarkable property that was proven in~\cite{GT2020a}.
\begin{theorem*}[\cite{GT2020a}]
Let $\p,\q\in\md(X)$ and set $a\eqdef|X|$. Suppose that the components of $\q$ are positive and rational. That is, there exists $n_1,...,n_a\in\mbb{N}$ such that
\be
\q=\left(\frac{n_1}{n},...,\frac{n_a}{n}\right)\quad,\quad n\eqdef\sum_{x=1}^{a}n_x\;.
\ee
Let
\be
\r\eqdef\bigoplus_{x=1}^{a}p_x\u^{(n_{x})}\in\md(n)\;,
\ee
where each $\u^{(n_x)}$ is $n_x$-dimensional uniform distribution.
Then,
\be\label{simp}
(\p,\q)\sim(\r,\u)
\ee
where $\u\in\md(n)$ is the $n$-dimensional uniform distribution.
\end{theorem*}

\section{Channel Divergences}

{ Quantum divergences have numerous applications in quantum information theory~\cite{Wilde2013} and quantum resource theories~\cite{CG2019}. As discussed in the preliminary section, their defining property is the DPI. Therefore, to define a channel divergence one has to extend the DPI to the channel domain. Since the most general and physically realizable operation that can be applied to a quantum channel is a superchannel, we will define channel divergences in terms of superchannels.}

We first start with the formal definition of a channel divergence and a channel relative entropy.
We will use the notation $\sD$ for a general channel divergence, to distinguish it from a general quantum divergence $\D$, or a general classical divergence $\xD$.

\begin{definition}\label{qdefd}
Let $$\sD:\bigcup_{A,B}\Big\{\cptp(A\to B)\times\cptp(A\to B)\Big\}\to\mbb{R}\cup\{\infty\}$$ be a function acting on pairs of quantum channels in finite dimensions.
\begin{enumerate}
\item The function $\sD$ is called a \emph{channel divergence}  if it satisfies the generalized Data Processing Inequality (DPI). That is, for any $\mM,\mN\in\cptp(A\to B)$ and a superchannel $\Theta\in\super(AB\to A'B')$
\be\nonumber
\sD\big(\Theta[\mM]\big\|\Theta[\mN]\big)\leq \sD(\mM\|\mN)\;.
\ee
\item A channel divergence $\sD$ is called a \emph{channel relative entropy} if it satisfies the following additional properties:
\begin{enumerate}
\item Additivity. For any $\mM_1,\mM_2\in\cptp(A\to B)$ and any $\mN_1, \mN_2\in\cptp(A'\to B')$
\be\nonumber
\sD\left(\mM_1\otimes\mM_2\big\|\mN_1\otimes\mN_2\right)= \sD(\mM_1\|\mN_1)+\sD(\mM_2\|\mN_2)\;.
\ee
\item Normalization. 
{
\be\nonumber
\sD\left(\begin{bmatrix} 1& 0\\ 0 & 0\end{bmatrix}\;\Big\|\; \begin{bmatrix} 1/2 & 0\\ 0 &1/2\end{bmatrix}\right)=1\;.
\ee
where quantum states are viewed as channels with one dimensional input.}
\end{enumerate}
\end{enumerate}
\end{definition}

\begin{remark}
The term ``channel relative entropy" is sometimes referred in literature to a quantity that extends the Umegaki relative entropy as in ~\eqref{qdsup}. However, in~\cite{Fang2020} it was shown not to be additive under tensor products and therefore does not fit with our definition above of a relative entropy. Moreover, its regularized version, which was shown in~\cite{Fang2020} to be equal to the amortized divergence (and which is known to be additive~\cite{Fang2020b}), is the one that appears in most applications. We therefore choose to call the expression in ~\eqref{qdsup}  a channel \emph{divergence} rather than a channel relative entropy.
\end{remark}

{ As an example, consider $D_{\max}(\mM\|\mN)$ as defined in~\eqref{max1}.
This channel divergence has several operational interpretations in quantum information~\cite{Kaur-2018a,Gour2018b,LW2019,LY2020,GW2019,Wang2019,Wang2019b,GS2020,Bauml2020,Fang2019,Fang2020,Katariya2020}. One can easily verify that it is additive and satisfies the normalization condition of a relative entropy. It also satisfies the DPI as defined above. This can be seen from the fact that if $t\mN-\mM$ is a CP map, then also $\Theta\left[\mN-t\mM\right]=\Theta[\mN]-t\Theta[\mM]$ is a CP map, so that $D_{\max}(\Theta[\mM]\|\Theta[\mN])\leq D_{\max}(\mM\|\mN)$. Therefore, according to the definition above, $D_{\max}(\mM\|\mN)$ is a channel relative entropy.

Another important example is the hypothesis testing divergence $D_{\min}^{\epsilon}(\mM\|\mN)$ as defined in~\eqref{min2}. This divergence plays a key role in hypothesis testing with quantum channels~\cite{Wang2019}, although it is not additive (unless $\eps=0$) which means that it is not a relative entropy according to our definition. Still it satisfies the DPI for channels~\cite{Gour2019}. 

Also the diamond norm is a channel divergence according to our definition above since it satisfies the channel DPI~\cite{Gour2019}, and it has an operational interpretation in terms of the optimal success probability of distinguishing between two channels (see e.g.~\cite{Watrous2018}). This again supports our definition of channel divergences, particularly the extension of the DPI with superchannels. This channel version of the DPI is satisfied by channel divergences that has physical and operational interpretations. We will discuss more examples of such channel divergences with operational interpretation  in the concluding section.
}

Channel divergences can be viewed as a generalization of quantum divergences since the set $\cptp(A\to B)$ with $|A|=1$ can be viewed as the set of quantum states $\md(B)$. Hence, we view here the quantum states in $\md(B)$ as a special type of  quantum channels, i.e. those in $\cptp(1\to B)$ where $1$ represents the one dimensional trivial system. However, there is another type of quantum channels that can be identified with quantum states. These are the replacement channels. Let $\sigma\in\md(B)$ and define the channel $\mR_\sigma\in\cptp(A\to B)$ with $|A|>1$ as 
\be\label{replace}
\mR_\sigma(\omega^A)\eqdef\tr[\omega^A]\sigma^B\quad\quad\forall\;\omega\in\ml(A)\;.
\ee
Therefore, this channel is uniquely determined by the dimension of system $A$ and the state $\sigma^B$. It is therefore natural to ask if channel divergences between replacement channels reduce to quantum divergences between the states that define the 
replacement channels. Not too surprising, we will see below that the answer to this question is on the affirmative.

\subsection{Basic Properties}

Channels divergences and relative entropies borrow some of their properties from quantum divergences. In this section we discuss a few of these basic properties. { The significance of these properties is that they are satisfied by all channel divergences (or relative entropies) and therefore are not unique to a particular divergence. Later on we will see that some divergences have additional unique properties that do not satisfied by all divergences.} We will say that a channel divergence $\sD$ is faithful if $\sD(\mN\|\mM)=0$ implies that $\mM=\mN$.

\begin{theorem}[Properties of Channel Divergences]\label{properties}
Let $\sD$ be a channel divergence. Then,
\begin{enumerate}
\item If $\sD(1\|1)=0$ (here 1 stands for the trivial channel in $\cptp(1\to 1)$) then for any two channels $\mM,\mN\in\cptp(A\to B)$ 
\be
\sD(\mM\|\mN)\geq 0\;,
\ee
with equality if $\mM=\mN$.
\item If $\sD$ is a channel relative entropy then $\sD(1\|1)=0$.
\item $\sD$ is faithful if and only if its reduction to classical states (i.e. probability vectors)  is faithful.
\item For any two replacement channels, $\mR_{\sigma_1},\mR_{\sigma_2}\in\cptp(A\to B)$, as defined in~\eqref{replace} with $\sigma_1,\sigma_2\in\md(B)$ and $|A|>1$
\be
\sD\left(\mR_{\sigma_1}\big\|\mR_{\sigma_2}\right)=\sD(\sigma_1\|\sigma_2)\;.
\ee  
\item Suppose $\sD$ is a channel relative entropy, and let $\mE_1,...,\mE_n\in\cptp(A\to B)$ be a set of $n$ orthogonal quantum channels; i.e. their Choi matrices satisfies $\tr[J_{\mE_j}^{AB}J_{\mE_k}^{AB}]=0$ for all $j\neq k\in[n]$. Then, for any probability vector $\p=\{p_x\}_{x=1}^n$, and $\mN\eqdef\sum_{x=1}^np_x\mE_x$,  
\be
\sD\left(\mE_x\big\|\mN\right)=-\log(p_x)\quad\forall\;x=1,...,n.
\ee
\item If $\sD$ is a channel relative entropy then
for all $\mM,\mN\in\cptp(A\to B)$
\be
D_{\min}(\mM\|\mN)\leq \sD(\mM\|\mN)\leq D_{\max}(\mM\|\mN)\;,
\ee
where $D_{\min}$ and $D_{\max}$ are the channel min and max relative entropies as defined in~\eqref{min2} and~\eqref{max1}, respectively.
\item Let $\sD$ be a channel relative entropy, $\mR\in\cptp(A\to B)$ be the completely randomizing channel (i.e.\ $\mR=\mR_{\u^B}$ meaning $\mR(\rho^A)=\u^B$ is the maximally mixed state $\u^B$ for all $\rho\in\md(A)$), and $\mV\in\cptp({A\to B})$ be an isometry channel (we assume $|A|\leq |B|$). Then,
\be
\sD\left(\mV^{A\to B}\big\|\mR^{A\to B}\right)=\log|AB|\;.
\ee
\item If $\sD$ is a channel relative entropy then for any $\mN,\mM,\mE\in\cptp(A\to B)$ 
\be\label{triang}
\sD(\mN\|\mM)\leq\sD(\mN\|\mE)+D_{\max}(\mE\|\mM)\;.
\ee
\item If $\sD$ is a channel relative entropy then for any $\mN, \mE, \mM \in\cptp(A\to B)$ 
	\ba\label{conti}
		&\sD(\mN\|\mM) - \sD(\mE\|\mM)\\ 
		&\leq \min_{0\leq s\leq 2^{-D_{\max}(\mE\|\mN)}} D_{\max}\left(\mN+s(\mM-\mE)\big\|\mM\right)
	\ea
	Moreover, if $J_{\mN},\;J_{\mE},$ and $J_{\mM}$ have full support then
		\begin{align}\label{bound0}
	\sD({\mN}\|{\mM}) - \sD({\mE}\|{\mM})\leq \log \left( 1 + \frac{  \| J_{\mN} - J_{\mE} \|_{\infty} } {\lambda_{\min}(J_{\mE}) \lambda_{\min}(J_{\mM}) }  \right)\;.
	\end{align}
\end{enumerate}
\end{theorem}

\begin{remark}
Property 8 implies that for any $\mN,\mE,\mF\in\cptp(A\to B)$
\be
\big|\sD(\mN\|\mE)-\sD(\mN\|\mF)\big|\leq D_{T}(\mE\|\mF)
\ee
where
\be
D_{T}(\mE\|\mF)\eqdef\max\big\{D_{\max}(\mE\|\mF),D_{\max}(\mF\|\mE)\big\}
\ee
is a metric (whose state version is known as the Thompson metric) on $\cptp(A\to B)$. Hence, in particular, any channel relative entropy $\sD$ is continuous in its second argument on the subset of $\cptp(A\to B)$ consisting of channels that have strictly positive Choi matrices. Similarly, Property 9 implies a continuity in the first argument of $\sD$. That is, if $\mE$ is very close to $\mN$ then $s$ can be taken to be very close to one as long as $\supp(J_{\mE})\subseteq\supp(J_{\mN})$.  Note that in this case, if we also have $\supp(J_{\mN})\subseteq\supp(J_{\mM})$
then the continuity of $D_{\max}$ implies that $D_{\max}\left(\mN+s(\mM-\mE)\big\|\mM\right)$ goes to zero as $s$ goes to one
(recall that if $s=1$ then $\mE=\mN$).
\end{remark}

\subsection{Divergences of Classical Channels}\label{classical}

The properties discussed in the previous subsection applies to all quantum-channel divergences. { Here we show that there are additional properties satisfied by channel relative entropies, when the inputs are restricted to be classical. That is, in this subsection we consider classical dynamical divergences; i.e. divergences of classical channels. Most interestingly,  we will see that there is only one classical-channel relative entropy that reduces to the KL divergence.}
 We start with the following theorem.

\begin{theorem}\label{thm2}
Let $\sD$ be a classical channel divergence that reduces to the classical (state) divergence $\xD$ on classical states in $\md(X)\times\mD(X)$. Suppose further that $\xD$ is quasi-convex. Then, for all classical channels $\mM,\mN\in\cptp(X\to Y)$
\be
\underline{\xD}(\mM\|\mN)\leq\sD(\mM\|\mN)\leq\overline{\xD}(\mM\|\mN)
\ee
where
\ba
&\underline{\xD}(\mM\|\mN)\eqdef\max_{x\in\{1,...,|X|\}}\xD\left(\mM(|x\lr x|)\big\|\mN(|x\lr x|)\right)\\
&\overline{\xD}(\mM\|\mN)\eqdef\inf_{|Z|\in\mbb{N}}\xD\left(\p\|\q\right)
\ea
where the infimum is over all $\p,\q\in\md(Z)$ that satisfies
$(\p,\q)\succ\big(\mM(|x\lr x|),\mN(|x\lr x|)\big)$ for all $\;x\in[|X|]$.
Moreover, $\underline{\xD}$ is a classical channel relative entropy, and $\overline{\xD}$ is a normalized classical channel divergence.
\end{theorem}

\begin{remark}
In the next section, we will develop a general framework to extend channel divergences from one domain to a larger one, and
the optimality of the divergences $\underline{\xD}$ and $\overline{\xD}$ will follow trivially from that general formalism. Hence, the theorem above can be viewed as a corollary of the third property in Theorem~\ref{outs} of the next section. Moreover, we will see shortly from the closed formula for $\overline{\xD}(\mM\|\mN)$ that w.l.o.g. we can bound $|Z|\leq |Y|$. 
\end{remark}

To illustrate the expression above for $\overline{\xD}$, in Fig.~\ref{3curves} we draw three lower Lorenz curves associated with the two channels $\mN,\mM\in\cptp(X\to Y)$, where $|X|=2$ and $|Y|=4$. In this example, the channel $\mM$ is defined via the probability vectors $\m_0\eqdef\mM(|0\lr 0|)=(1/3,1/4,1/4,1/6)$ and $\m_1\eqdef\mM(|1\lr 1|)=(5/12,1/6,1/4,1/6)$, and the channel $\mN$ via $\n_0\eqdef\mN(|0\lr 0|)=(1/12,1/6,1/3,5/12)$
and $\n_1\eqdef\mN(|1\lr 1|)=(1/12,1/12,1/2,1/3)$. The figure helps to see how to compute the optimal pair $(\p,\q)$ that satisfies 
both $(\p,\q)\succ(\m_0,\n_0)$ and $(\p,\q)\succ(\m_1,\n_1)$. The Lorenz curve, $\mL(\p,\q)$, must be the closest possible curve to both $\mL(\m_0,\n_0)$ and $\mL(\m_1,\n_1)$ and it must also be below both of them. In Fig.~\ref{3curves}, the dashed line is the Lorenz curve of $\mL(\p,\q)$. Its vertices are $(0,0)$, $(5/12,1/12)$, $(7/12,2/12)$, $(5/6,7/12)$ and $(1,1)$. This reveals that the optimal $\p$ and $\q$ are given by $\p=(5/12,1/6,1/4,1/6)$ and $\q=(1/12,1/12,5/12,5/12)$.

\begin{figure}[h]\centering
    \includegraphics[width=0.3\textwidth]{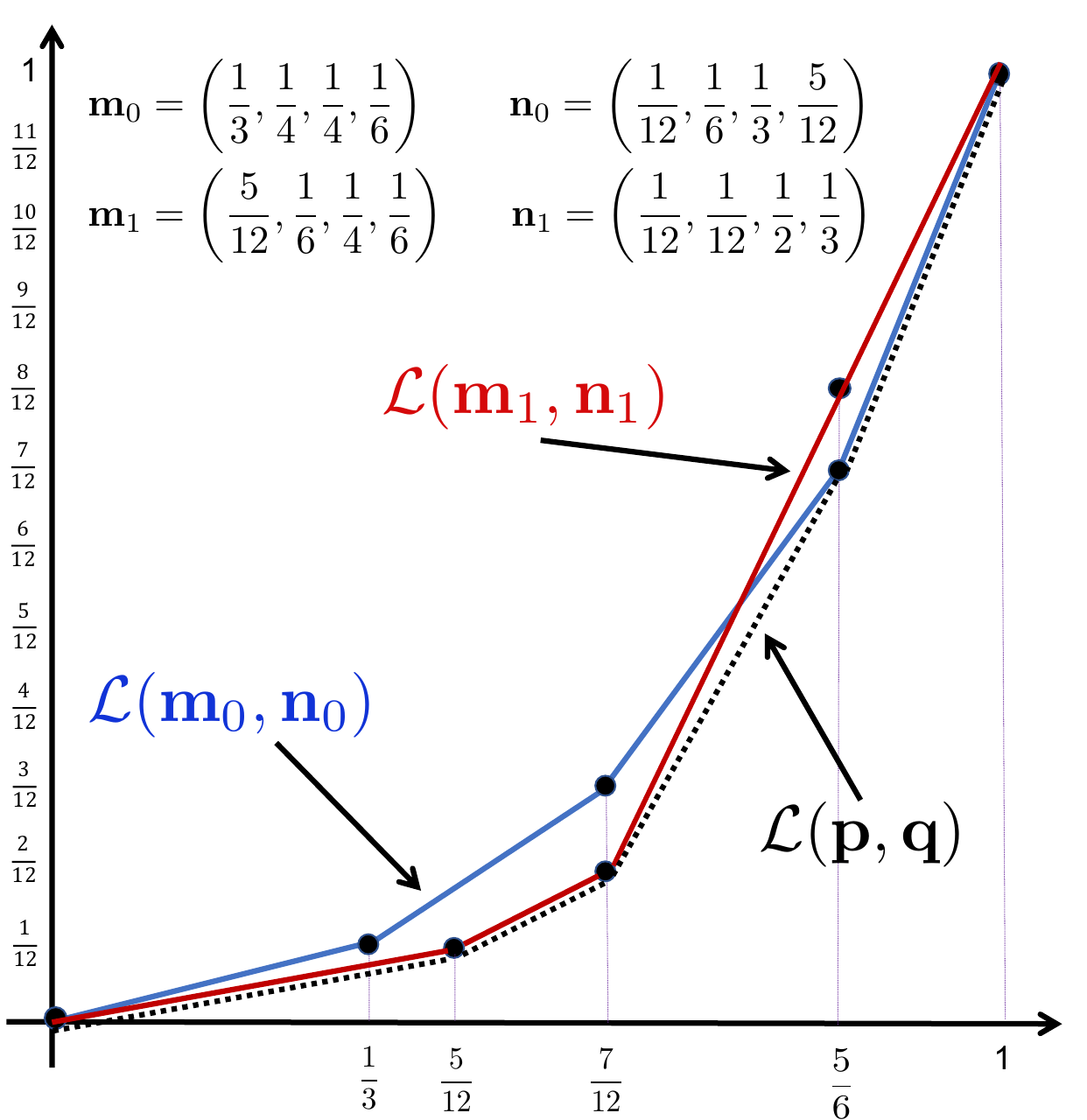}
  \caption{\linespread{1}\selectfont{\small Three Lower Lorenz Curves. The channels $\mM,\mN\in\cptp(X\to Y)$ with $|X|=2$ and $|Y|=4$.}}
  \label{3curves}
\end{figure}

For the general case, we provide in the theorem below a closed formula for $\overline{\xD}(\mM\|\mN)$ for any two classical channels $\mM,\mN\in\cptp(X\to Y)$ in finite dimensions $|X|,|Y|<\infty$. For any $x=1,...,|X|$, we will denote by $(\m_x,\n_x)\eqdef\big(\mM(|x\lr x|),\mN(|x\lr x|)\big)$
and by $M=(M_{y|x})$ and $N=(N_{y|x})$ the $|Y|\times |X|$ matrices whose components are $M_{y|x}=\la y|\m_x|y\ra$ and $N_{y|x}=\la y|\n_x|y\ra$, respectively. We also rearrange the components of the columns of $M$ and $N$ such that for each $x=1,...,|X|$
\be
\frac{M_{1|x}}{N_{1|x}}\geq\frac{M_{2|x}}{N_{2|x}}\geq\cdots\geq\frac{M_{|Y||x}}{N_{|Y||x}}\;.
\ee
Therefore, aside from the points $(0,0)$ and $(1,1)$, for any $x=1,...,|X|$, the vertices of the lower Lorenz curve of $\mL(\m_x,\n_x)$ are given by $\{(a_{zx},b_{zx})\}_{z=1}^{|Y|-1}$, where
\be
a_{zx}\eqdef\sum_{y=1}^{z}M_{y|x}\quad\text{and}\quad b_{zx}\eqdef \sum_{y=1}^{z}N_{y|x}\;.
\ee 

\begin{theorem}[Closed Formula]\label{thm:cc}
Let $\D$ be a classical divergences and let $\mN,\mM\in\cptp(X\to Y)$. Using the notations above, the maximal classical channel extension $\overline{\xD}$ is given by
\be
\overline{\xD}(\mM\|\mN)=\xD(\p\|\q)\;,
\ee
where $\p=\{p_z\}_{z=1}^{|Y|}$ and $\q=\{q_z\}_{z=1}^{|Y|}$ are $|Y|$-dimensional probability vectors given by
\ba
&p_1=M_{1x_1}\quad{,}\quad p_z=a_{zx_z}-a_{(z-1)x_{z-1}}\\
&q_1=N_{1x_1}\quad{,}\quad q_z=b_{zx_z}-b_{(z-1)x_{z-1}}\quad\forall\;z=2,...,|Y|
\ea
where $x_1,...,x_{|Y|}\in\{1,...,|X|\}$ are defined by induction via the relations
\be
\frac{N_{1x_1}}{M_{1x_1}}=\min_{x\in[|X|]}\frac{N_{1x}}{M_{1x}}\;,
\ee
and for any $z=2,...,|Y|$
\be\label{defk}
\frac{b_{zx_z}-b_{(z-1)x_{z-1}}}{a_{zx_z}-a_{(z-1)x_{z-1}}}=\min_{x\in\{1,...,|X|\}}\frac{b_{zx}-b_{(z-1)x_{z-1}}}{a_{zx}-a_{(z-1)x_{z-1}}}\;.
\ee
\end{theorem}

The maximal divergence $\overline{\xD}$ is in general not additive, and therefore, is not a relative entropy. 
In the next result, we prove that if $\sD$ is a channel relative entropy that reduces to the KL relative entropy, then it is unique.

\begin{myt}{\color{yellow} Uniqueness of the Channel KL-Relative Entropy}
\begin{theorem}\label{mainresult}
Let $\sD$ be a classical channel divergence that reduces to the Kullback-Leibler divergence, $D$, on classical states. If $\sD$ is continuous in its second argument then for all $\mN,\mM\in\cptp(X\to Y)$
\be
\sD(\mM\|\mN)=\max_{x\in\{1,...,|X|\}}D\left(\mM(|x\lr x|)\big\|\mN(|x\lr x|)\right)\;.
\ee
\end{theorem}
\end{myt}

The uniqueness theorem above holds only for the KL divergence and it is not clear to the author if this uniqueness still holds for the case that the KL divergence is replaced with the R\'enyi divergences. In~\cite{GT2020b} (cf.~\cite{Matsumoto2018b,WGE2017}) it was shown that the Umegaki relative entropy (and in particular the KL-divergence) is the only asymptotically continuous divergence. Therefore, the uniqueness theorem above implies that there is only one classical channel divergence that on classical states is asymptotically continuous.

\section{Optimal Extensions}

In this section we apply the extension techniques developed in~\cite{GT2020b} for general resource theories to study the optimal extensions of a classical divergence to a channel divergence. { We point out that the extensions we consider are optimal in the sense that they are either minimal or maximal among all possible extensions. For example, if $\sD(\mN\|\mM)$ is a channel divergence that reduces to a classical state divergence $\xD(\p\|\q)$ when $\mN=\p$ and $\mM=\q$ are classical states, then $\sD(\mN\|\mM)$ is bounded from below by the minimal extension of $\xD(\p\|\q)$ to the channel domain, and bounded from above by the maximal extension of $\xD(\p\|\q)$ to the channel domain.}
One can also consider extensions of quantum divergences to channel divergences. As we will see, such extensions give rise to additional types of channel divergences. We start with the general framework for channel extensions of divergences.

\subsection{General Framework for Extensions}

In the following theorem we apply the results that were given in~\cite{GT2020b} for a general resource theory, to channel-extensions of divergences. We start with the definition of an $\mr$-divergence.
\begin{definition}\label{rdiv}
Let $\mr(A\to B)\subset\cptp(A\to B)$ be a subset of quantum channels for any two physical systems $A$ and $B$. A function $$\C:\bigcup_{A,B}\mr(A\to B)\times\mr(A\to B)\to\mbb{R}_{+}$$ is called an $\mr$-divergence if for any $\mM,\mN\in\mr(A\to B)$ and any superchannel $\Theta\in\super(AB\to A'B')$ such that $\Theta[\mM]\in\mr(A'\to B')$ and $\Theta[\mN]\in\mr(A'\to B')$
\be
\C\big(\Theta[\mM]\big\|\Theta[\mN]\big)\leq \C(\mM\|\mN)\;.
\ee
\end{definition}

Any $\mr$-divergence has two optimal extensions to a quantum channel divergence: 
\begin{enumerate}
\item The minimal channel-extension, for any $\mM,\mN\in\cptp(A\to B)$
\be\label{min}
\underline{\C}(\mM\|\mN)\eqdef\sup\C(\Theta[\mM]\|\Theta[\mN])
\ee
where the supremum is over all systems $A',B'$ and all
$\Theta\in\super(AB\to A'B')$ such that $\Theta[\mM],\Theta[\mN]\in\mr(A'\to B')$.
\item The maximal channel-extension, for any $\mM,\mN\in\cptp(A\to B)$
\be\label{max}
\overline{\C}(\mM\|\mN)\eqdef\inf\C(\mE\|\mF)
\ee
where the infimum is also over all systems $A',B'$, and all $\mE,\mF\in\mr(A'\to B')$ such that there exists $\Theta\in\super(A'B'\to AB)$ that satisfies $\mM=\Theta[\mE]$ and $\mN=\Theta[\mF]$.
\end{enumerate}

\begin{remark}
In the definition above we assumed that $\mr(A\to B)$ is a subset of quantum channels. By taking $\mr(A\to B)$ to be a subset of replacement channels as in~\eqref{replace}, the extensions above can also be applied to state divergences.
\end{remark}

\begin{theorem}\label{outs}
 Let $\mr(A\to B)\subset\cptp(A\to B)$, and let $\C$ be an $\mr$-divergence. Then, its maximal and minimal channel-extensions $\overline{C}$ and $\underline{C}$ have the following  properties:
\begin{enumerate} 
\item \textbf{Reduction.} For any $\mM,\mN\in\mr(A\to B)$
\be
\underline{\C}(\mM\|\mN)=\overline{\C}(\mM\|\mN)={\C}(\mM\|\mN)\;.
\ee
\item \textbf{Data Processing Inequality.} For any $\mM,\mN\in\cptp(A\to B)$ and any $\Theta\in\super(AB\to A'B')$
\ba
&\underline{\C}\big(\Theta[\mM]\big\|\Theta[\mN]\big)\leq\underline{\C}(\mM\|\mN)\;\text{and},\\
&\overline{\C}\big(\Theta[\mM]\big\|\Theta[\mN]\big)\leq\overline{\C}(\mM\|\mN)\;.
\ea
\item \textbf{Optimality.} Any quantum channel divergence $\sD$ that reduces to $\C$ on pairs of channels in $\mr(A\to B)$, must satisfy for all $\mM,\mN\in\cptp(A\to B)$
\be\label{bounds9}
\underline{\C}(\mM\|\mN)\leq \sD(\mM\|\mN)\leq\overline{\C}(\mM\|\mN)\;.
\ee
\item \textbf{Sub/Super Additivity.} Suppose $\C$ is weakly additive;  that is,
for any $k\in\mbb{N}$
\be
\C(\mM^{\otimes k}\|\mN^{\otimes k})=k\C(\mM\|\mN)
\ee
Then, $\underline{\C}$ is super-additive and $\overline{\C}$ is sub-additve. Explicitly, for any $\mM_1,\mM_2\in\cptp(A\to B)$ and any $\mN_1, \mN_2\in\cptp(A'\to B')$
\ba
&\underline{\C}\left(\mM_1\otimes\mM_2\big\|\mN_1\otimes\mN_2\right)\geq \underline{\C}(\mM_1\|\mN_1)+\underline{\C}(\mM_2\|\mN_2)\\
&\overline{\C}\left(\mM_1\otimes\mM_2\big\|\mN_1\otimes\mN_2\right)\leq \overline{\C}(\mM_1\|\mN_1)+\overline{\C}(\mM_2\|\mN_2)\;.
\ea
\item \textbf{Regularization.} If $\C$ is weakly additive under tensor products then any weakly additive quantum channel divergence $\sD$ that reduces to $\C$ on pairs of channels in $\mr(A\to B)$, must satisfy for all $\mM,\mN\in\cptp(A\to B)$
\be\label{bounds10}
\underline{\C}^\reg(\mM\|\mN)\leq \sD(\mM\|\mN)\leq\overline{\C}^\reg(\mM\|\mN)\;,
\ee
where
\ba\label{xzx}
&\underline{\C}^\reg(\mM\|\mN)=\lim_{n\to\infty}\frac1n\underline{\C}\left(\mM^{\otimes n}\big\|\mN^{\otimes n}\right)\;\text{and},\\
&\overline{\C}^\reg(\mM\|\mN)=\lim_{n\to\infty}\frac1n\overline{\C}\left(\mM^{\otimes n}\big\|\mN^{\otimes n}\right)\;,
\ea
and $\underline{\C}^{\reg}$ and $\overline{\C}^{\reg}$ are themselves weakly additive channel divergences.
\end{enumerate}
\end{theorem}
\begin{remark}
In the case that $\C$ is additive (even weakly additive) the minimal and maximal extensions are super-additive and sub-additive, respectively. This, in turn, implies that the limits in~\eqref{xzx} exists so that $\overline{\C}^\reg$ and $\underline{\C}^\reg$ are well defined. Moreover, in general, the bounds on $\sD$ in~\eqref{bounds10} are tighter than the bounds in~\eqref{bounds9}. This assertion follows from the fact that that $\underline{\C}$ is super-additive and in particular satisfies $\underline{\C}^{\reg}(\mN\|\mM)\geq \underline{\C}(\mN\|\mM)$. Similarly, the sub-additivity of $\overline{\C}$ implies that $\overline{\C}^{\reg}(\mN\|\mM)\leq \overline{\C}(\mN\|\mM)$.
\end{remark}

In the following subsections we apply Theorem~\ref{outs} to the cases that $\mr$ is the subset of all quantum states (i.e. replacement channels) and the subset of all classical states. We will see that this give rise to several optimal channel-extensions of state divergences. We start, however, by using the theorem above to prove the uniqueness of the max channel relative entropy.

\subsection{Uniqueness of the max relative entropy}

The max relative entropy is defined for any $\rho,\sigma\in\md(A)$ as
\be
D_{\max}(\rho\|\sigma)\eqdef\log\min\Big\{t\;:\;t\sigma\geq\rho\quad t\in\mbb{R}\Big\}\;.
\ee
The max relative entropy is unique with respect to its monotonicity property. Unlike the relative entropy and all the other R\'enyi entropies, it behaves monotonically under any CP map (not necessarily trace preserving or trace non-increasing). More precisely, let $\rho,\sigma\in\md(A)$, and let $\mE\in\cp(A\to B)$ be such that
$\mE(\rho)$ and $\mE(\sigma)$ are normalized quantum states in $\md(B)$. Since we do not assume here that $\mE$ is trace non-increasing we cannot conclude that there exists a CPTP map that achieves the same task; i.e. taking the pair $(\rho,\sigma)$ to the pair $(\mE(\rho),\mE(\sigma))$. Yet, the max divergence behaves monotonically under such maps; explicitly, for a given $\rho,\sigma\in\md(A)$,
\be
 D_{\max}\big(\mE(\rho)\|\mE(\sigma)\big)\leq D_{\max}(\rho\|\sigma)\;,
\ee
for any $\mE\in\cp(A\to B)$ for which  $\mE(\rho),\mE(\sigma)\in\md(B)$ for the given $\rho$ and $\sigma$. Note that we assumed here that $D_{\max}$ is defined only on pairs of normalized states. Extensions to subnormalized states can be made using the techniques studied in~\cite{GT2020b}.

For quantum channels, $D_{\max}$ has been defined analogously to the states case as~\cite{LW2019}
\be
D_{\max}(\mN\|\mM)\eqdef\log\min\Big\{t\in\mbb{R}\;:\;t\mM\geq\mN\Big\}\;.
\ee
One can easily see that similar to the states case, $D_{\max}(\mN\|\mM)$ behaves monotonically under any CP preserving (CPP) supermap that takes the pair of channels $(\mM,\mN)$ to any other pair of channels $(\mM',\mN')$. Specifically, let $\Theta\in{\rm CPP}(AB\to A'B')$ be a CPP supermap that is not necessarily a superchannel, and suppose that $\mM'\eqdef\Theta[\mN]$ and $\mN'\eqdef\Theta[\mN]$ are quantum channels in $\cptp(A'\to B')$. Then,
\be
D_{\max}(\mM'\|\mN')\leq D_{\max}(\mM\|\mN)\;.
\ee

We show here that the extension of $D_{\max}$ from classical states to quantum channels is unique.
We already know from Property 6 of Theorem~\ref{properties} that any channel relative entropy cannot exceed $D_{\max}$.
In fact, in the proof of Property 6 of Theorem~\ref{properties} we only use the normalization property of relative entropy. Therefore, one can conclude something slightly stronger that it is even not possible to extend $D_{\max}$  to a non-additive channel divergence. 

From the optimality property of Theorem~\ref{outs} it follows that in order to prove uniqueness, it is sufficient to show that the maximal and minimal extensions of $D_{\max}$ are equal to each other. Applying the general framework for extensions developed in the previous subsection, the maximal and minimal extensions of $D_{\max}$ to quantum channels ,denoted by $\bD_{\max}$ and $\uD_{\max}$, respectively, are defined for any $\mN,\mM\in\cptp(A\to B)$ as
\ba\label{dmax1}
&\underbar{D}_{\max}(\mN\|\mM)\eqdef\sup_{\substack{|X|\in\mbb{N}\\ \Theta\in\super(AB\to X)}} D_{\max}\Big(\Theta[\mN]\big\|\Theta[\mM]\Big)\;;\\
&\bD_{\max}(\mN\|\mM)\eqdef\inf_{\substack{|X|\in\mbb{N}\;,\;\p,\q\in\md(X)\\ \mN=\Theta[\p]\;,\;\mM=\Theta[\q]\\\Theta\in\super(X\to AB)}}D_{\max}(\p\|\q)\;.
\ea

\begin{theorem}[Uniqueness of the dynamical max relative entropy]\label{maxu}
Let $\sD$ be a channel divergence that reduces to $D_{\max}$ on classical probability distributions; i.e. for any classical system $X$ and $\p,\q\in\md(X)$, $\sD(\p\|\q)=D_{\max}(\p\|\q)$. Then, for all $\mN,\mM\in\cptp(A\to B)$
\be
\sD(\mN\|\mM)=D_{\max}(\mN\|\mM)\;.
\ee
\end{theorem}
\begin{remark}
Note that we do not assume that $\sD$ is a channel relative entropy (i.e. additive), only a divergence that reduces to the max relative entropy on classical states. The main idea of the proof is to show that the two expressions in~\eqref{dmax1} are both equal to $D_{\max}$.
\end{remark}

\subsection{Extension from a quantum (state) divergence to a channel divergence}

In this section we study the optimal extensions of quantum state divergences to quantum channel divergences.
For a given quantum state divergence $\D$ we denote by $\ubd$ its minimal channel-extension. According to~\eqref{min} the minimal channel-extension is given by
\begin{align}
&\ubd(\mN\|\mM)\eqdef\sup_{\Theta\in\super(AB\to R')}\D\big(\Theta\left[\mN\right]\big\|\Theta[\mM]\big)=\sup_{\substack{\mE\in\cptp\\\psi\in\md(RA)}}\nonumber\\
& \D\Big(\mE^{BR\to R'}\circ\mN^{A\to B}(\psi^{AR})\Big\|\mE^{BR\to R'}\circ\mM^{A\to B}(\psi^{AR})\Big)\nonumber\\
&=\max_{\psi\in\md(RA)}\D\Big(\mN^{A\to B}(\psi^{AR})\big\|\mM^{A\to B}(\psi^{AR})\Big)\label{qdsup}
\end{align}
where in the last equality the supremum has been replaced with a maximum since w.l.o.g. we can assume that $|R|=|A|$ and that $\psi^{RA}$ is a pure state~\cite{Cooney-2016a,Felix2018}. 

Similarly, by the definition in~\eqref{max}, the maximal channel-extension is given by
\begin{align}
&\bbd(\mN\|\mM)\\
&\eqdef\inf_{\substack{\Theta\in\super(R\to AB)\\\rho,\sigma\in\md(R)}} \Big\{\D(\rho\|\sigma)\;:\;\mN=\Theta[\rho],\;\mM=\Theta[\sigma]\Big\}\nonumber\\
&=\inf_{\substack{\mE\in\cptp(RA\to B)\\\rho,\sigma\in\md(R)}}\Big\{ \D(\rho\|\sigma)\;:\; \mN=\mE_\rho,\;\mM=\mE_\sigma\Big\}\label{117}
\end{align}
where for any density matrices $\rho,\sigma\in\md(R)$ and channel
$\mE\in\cptp(RA\to B)$ we denote
\ba\label{ep}
&\mE_\rho^{A\to B}(\omega^A)=\mE^{RA\to B}(\rho^{R}\otimes\omega^A)\quad\text{and},\\
&\mE^{A\to B}_\sigma(\omega^A)=\mE^{RA\to B}( \sigma^{R}\otimes\omega^A)\quad\forall\omega\in\ml(A)\;.
\ea
The channels above have been studied under the name \emph{environment-parametrized channels}~\cite{DW2019a,DW2019b,TW2016}, and the expression in~\eqref{qdsup} has been used in the literature for the cases that $\D$ is the trace norm (in which case $\ubd$ becomes the diamond norm~\cite{AKN1998}), Umegaki relative entropy, and quantum R\'enyi divergences~\cite{Cooney-2016a,Felix2018}. The following corollary is the restatement of Theorem~\ref{outs} for the optimal channel-extensions of quantum state divergences.

\begin{corollary}
Let $\D$ be a quantum (state) divergence, and let $\ubd$ and $\bbd$ be its minimal and maximal extensions to quantum channels. 
Then,
\begin{enumerate}
\item Both $\ubd$ and $\bbd$ are quantum-channel divergences.
\item Both $\ubd$ and $\bbd$ reduces to $\D$ on quantum states.
\item Any other channel divergences $\sD$ that reduces to $\D$ on quantum states must satisfy for all $\mN,\mM\in\cptp(A\to B)$
\be\label{ncad}
\ubd(\mN\|\mM)\leq \sD(\mN\|\mM)\leq\bbd(\mN\|\mM)\;.
\ee
\item If $\D$ is a weakly additive quantum state divergence then $\ubd$ is super-additive and $\bbd$ is sub-additive with respect to tensor products. Explicitly, for any $\mM_1,\mM_2\in\cptp(A\to B)$ and any $\mN_1, \mN_2\in\cptp(A'\to B')$
\ba
&\ubd\left(\mM_1\otimes\mM_2\big\|\mN_1\otimes\mN_2\right)\geq \ubd(\mM_1\|\mN_1)+\ubd(\mM_2\|\mN_2)\\
&\bbd\left(\mM_1\otimes\mM_2\big\|\mN_1\otimes\mN_2\right)\leq \bbd(\mM_1\|\mN_1)+\bbd(\mM_2\|\mN_2).
\ea
\item If $\D$ is a weakly additive quantum state divergence then any weakly additive quantum channel divergence $\sD$ that reduces to $\D$ on quantum states, must satisfy for all $\mM,\mN\in\cptp(A\to B)$
\be
\underline{\D}^\reg(\mM\|\mN)\leq \sD(\mM\|\mN)\leq\overline{\D}^\reg(\mM\|\mN)\;,
\ee
where
\ba
&\underline{\D}^\reg(\mM\|\mN)=\lim_{n\to\infty}\frac1n\underline{\D}\left(\mM^{\otimes n}\big\|\mN^{\otimes n}\right)\;\text{and},\\
&\overline{\D}^\reg(\mM\|\mN)=\lim_{n\to\infty}\frac1n\overline{\D}\left(\mM^{\otimes n}\big\|\mN^{\otimes n}\right)\;.
\ea
and $\ubd^{\reg}$ and $\bbd^{\reg}$ are themselves weakly additive normalized channel divergences.
\end{enumerate}
\end{corollary}

In addition to the corollary above, we have the following property for the maximal extension.
\begin{theorem}\label{lem1}
Let $\D$ be a jointly convex quantum divergence. Then, its maximal channel-extension $\bbd$ is also jointly convex.
\end{theorem}

The channel divergence $\ubd$  has been shown in~\cite{Gour2018b} to satisfy the generalized DPI.
For the case that $\D=D$ is the Umegaki relative entropy, it was shown in~\cite{Fang2020} that it satisfies a chain rule. The latter property in particular implies that its regularization can be expressed as~\cite{Fang2020}
\ba\label{amortized}
&\uD^{\reg}(\mN\|\mM)=\sup_{\rho,\sigma\in\md(RA)}\\
&\Big\{D\left(\mM^{A\to B}\left(\rho^{RA}\right)\big\|\mN^{A\to B}\left(\sigma^{RA}\right)\right)-D\left(\rho^{RA}\big\|\sigma^{RA}\right)\Big\}
\ea
where the expression on the RHS is known as the amortized divergence~\cite{Berta2018}. We will see below that any channel relative entropy that reduces to the Kullback-Leibler divergence on classical states must be no smaller than the above expression.

\subsection{Extensions from classical state divergences to channel divergences}

In this subsection we study optimal channel-extensions of a classical state divergence $\xD$. We define the following four optimal extensions of $\xD$ to quantum channel divergences.
\begin{definition} \label{subnp}
Let $\xD:\md(X)\times\md(X)\to \mbb{R}_+$ be a classical (state) divergence, and let $\mN,\mM\in\cptp(A\to B)$ be two quantum channels. We define four extensions of $\xD$ to quantum channels:
\begin{enumerate}
\item The minimal extension of $\xD$,   
\be
\underline{\xD}(\mN\|\mM)\eqdef\sup_{\substack{|X|\in\mbb{N}\\ \Theta\in\super(AB\to X)}} \xD\Big(\Theta\left[\mN\right]\big\|\Theta[\mM]\Big)\;.
\ee
\item The maximal extension of $\xD$,
\be\label{123g}
\overline{\xD}(\mN\|\mM)\eqdef\inf_{\substack{|X|\in\mbb{N}\;,\;\p,\q\in\md(X)\\ \mN=\Theta[\p]\;,\;\mM=\Theta[\q]\\\Theta\in\super(X\to AB)}}\xD(\p\|\q)\;.
\ee
\item The geometric extension of $\xD$,
\be\label{asd}
\widehat{\xD}(\mN\|\mM)\eqdef\sup_{\substack{|R|\in\mbb{N}\\ \Theta\in\super(AB\to R)}} \overline{\xD}\Big(\Theta\left[\mN\right]\big\|\Theta[\mM]\Big)
\ee
where $\overline{\xD}$ is the maximal quantum \emph{state}-extension of $\xD$. 
\item The min-max extension of $\xD$,
\be\label{mmce}
\underline{\xD}^{\uparrow}(\mN\|\mM)\eqdef\inf_{\substack{|R|\in\mbb{N}\;,\;\rho,\sigma\in\md(R)\\ \mN=\Theta[\rho]\;,\;\mM=\Theta[\sigma]\\\Theta\in\super(R\to AB)}}\underline{\xD}(\rho\|\sigma)\;,
\ee
where $\underline{\xD}$ is the minimal quantum \emph{state}-extension of $\xD$. 
\end{enumerate}
\end{definition}

From Theorem~\ref{outs} it follows that all the four functions above satisfy the generalized data processing inequality (and therefore they are indeed divergences), and they all reduce to the classical divergence $\xD$ on classical states. Moreover, for a pair of quantum states $\rho,\sigma\in\md(A)$ we have $\widehat{\xD}(\rho\|\sigma)=\overline{\xD}(\rho\|\sigma)$ and
$\underline{\xD}^{\uparrow}(\rho\|\sigma)=\underline{\xD}(\rho\|\sigma)$. In addition, Theorem~\ref{outs} implies that any channel divergence $\sD$, that reduces on classical states to a classical divergence $\xD$,  must satisfy for all $\mN,\mM\in\cptp(A\to B)$
\be
\underline{\xD}(\mN\|\mM)\leq \sD(\mN\|\mM)\leq\overline{\xD}(\mN\|\mM)\;.
\ee
This in particular applies to the cases $\sD=\widehat{\xD}$ and $\sD=\underline{\xD}^\uparrow$, and also note that in general, the optimal channel-extensions in~\eqref{qdsup} and~\eqref{117} of a quantum state divergence $\D$, that reduces to a classical state divergence $\xD$ on classical states, satisfy for all $\mM,\mN\in\cptp(A\to B)$
\be
\underline{\xD}(\mM\|\mN)\leq\ubd(\mM\|\mN)\leq\bbd(\mM\|\mN)\leq\overline{\xD}(\mM\|\mN)\;. 
\ee

\subsubsection{The minimal channel extension and its regularization}

Note that the minimal extension $\underline{\xD}$ can be expressed as
\ba\label{minimalex}
&\underline{\xD}(\mN\|\mM)\eqdef\\
&\sup \xD\Big(\mE_{BR\to X}\circ\mN_{A\to B}(\psi_{AR})\big\|\mE_{BR\to X}\circ\mM_{A\to B}(\psi_{AR})\Big)
\ea
where the supremum is over all systems $X,R$, over all $\psi\in\md(AR)$, and over all $\mE\in\cptp(BR\to X)$.
Note that w.l.o.g. we can assume that $\psi_{RA}$ is a pure state.
There is at least one divergence for which the expression above { coincides} with the optimization given in~\eqref{qdsup}.
\begin{theorem}\label{lem17}
Let $\mM,\mN\in\cptp(A\to B)$ and for any $\epsilon\in[0,1)$, let $\xD=D_{\min}^{\epsilon}$ be the classical hypothesis testing divergence { (i.e. the classical version of the one defined in~\eqref{ht}).} 
Then,
\ba
&\underline{\xD}(\mM\|\mN)=D_{\min}^{\epsilon}(\mM\|\mN)\\
&\eqdef\sup_{\psi\in\md(RA)}D_{\min}^{\epsilon}(\mM^{A\to B}(\psi^{RA})\|\mN^{A\to B}(\psi^{RA}))\;.
\ea
\end{theorem}
\begin{remark}
The above theorem implies that any channel divergence $\sD$ that reduces to $D_{\min}^{\epsilon}$ on \emph{classical} states satisfies for all $\mM,\mN\in\cptp(A\to B)$
\be
\sD(\mM\|\mN)\geq D_{\min}^{\epsilon}(\mM\|\mN)\;.
\ee 
\end{remark}

The minimal channel extension is typically not additive even if the classical divergence $\xD$ is additive (i.e. $\xD$ is a relative entropy).
However, from Theorem~\ref{outs} we know that if $\xD$ is a classical relative entropy then its channel extension $\underline{\xD}$ is super-additive. This means that the limit in its regularization exists;
\be
\underline{\xD}^{\reg}(\mN\|\mM)\eqdef\lim_{n\to \infty}\frac{1}{n}\underline{\xD}(\mN^{\otimes n}\|\mM^{\otimes n})\;.
\ee
From Theorem~\ref{outs} it follows that $\underline{\xD}^{\reg}$ is a weakly additive divergence, and all channel relative entropies
that reduces to $\xD$ on classical states must be no smaller than it.

Suppose now that $\xD=D$ is the KL divergence. In this case, we denote by $\uD^{\reg}$ its minimal regularized channel extension. Another closely related quantity that plays important role in applications is the regularized version of the minimal channel extension of the (quantum) Umegaki relative entropy, denoted as $D^\reg$. That is,
\be
D^{\reg}(\mM\|\mN)\eqdef\lim_{n\to \infty}\frac{1}{n}D(\mM^{\otimes n}\|\mN^{\otimes n})
\ee 
where 
\be
D(\mM\|\mN)\eqdef\sup_{\psi\in\md(RA)}D\Big(\mM^{A\to B}(\psi^{AR})\big\|\mN^{A\to B}(\psi^{AR})\Big)
\ee
with $D$ being the Umegaki relative  entropy. Since $\uD^\reg$ is the minimal channel extension we must have
\be\label{12345}
\uD^{\reg}(\mM\|\mN)\leq D^{\reg}(\mM\|\mN)\;.
\ee
Note that if $\mM$ and $\mN$ in the equation above are quantum states (i.e. the input dimension $|A|=1$) then the equality holds. This is due to the fact that that on quantum states, the Umegaki relative entropy $D^{\reg}(\rho\|\sigma)=D(\rho\|\sigma)$ equals to the minimal additive state extension of the KL relative entropy (see~\eqref{wer}). We now show that the equality also holds for any two channels.

\begin{theorem}\label{eqreg}
Let $\mM,\mN\in\cptp(A\to B)$ be two quantum channels. Then,
\be
\uD^{\reg}(\mM\|\mN)= D^{\reg}(\mM\|\mN)\;.
\ee
That is, $D^{\reg}(\mM\|\mN)$ is the smallest channel relative entropy that reduces to the KL relative entropy on classical states.
\end{theorem}

Starting with a classical relative entropy $\xD$, the process at which we arrived the weakly additive channel divergence $\underline{\xD}^{\reg}$ had 2 steps: (1) Extend $\xD$ to the minimal channel divergence $\underline{\xD}$, and (2) regularize $\underline{\xD}$ to obtain a weakly additive divergence. One can also introduce regularization in the state level and only then apply the channel extension. Specifically, starting with a classical relative entropy $\xD$, we apply the following four steps:
\begin{enumerate}
\item Extend $\xD$ to the minimal quantum state divergence $\underline{\xD}$.
\item Regularize $\underline{\xD}$ to get a weakly additive quantum state divergence $\underline{\xD}^{\reg}$.
\item Using the minimal extension, extend $\underline{\xD}^{\reg}$ to a channel divergence $\underline{\xD}_{ch}$.
\item Regularize $\underline{\xD}_{ch}$ to get a weakly additive channel divergence $\underline{\xD}_{ch}^{\reg}$.
\end{enumerate}
Fig.~\ref{orderex} illustrate these four steps. 

\begin{figure}[h]\centering
    \includegraphics[width=0.5\textwidth]{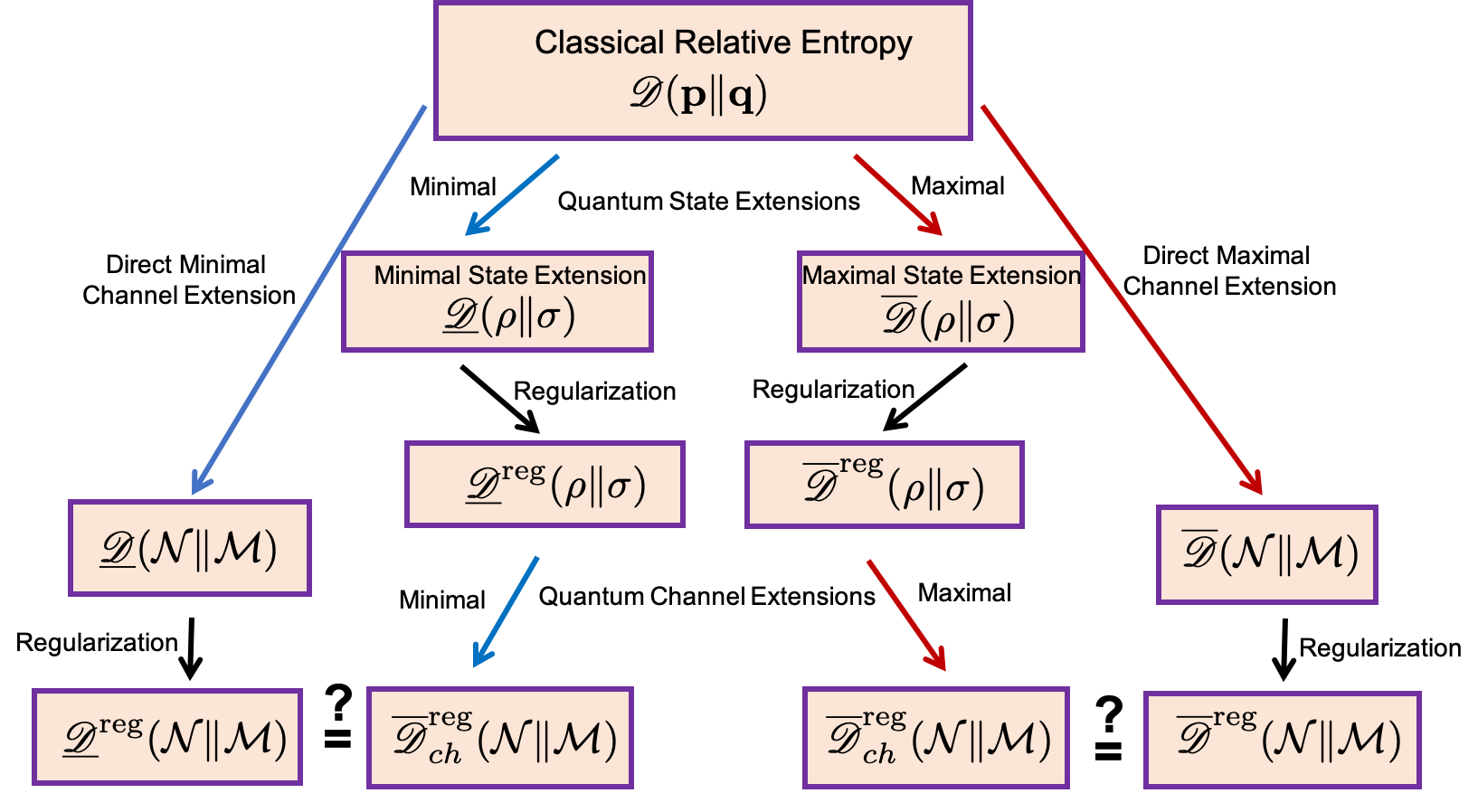}
  \caption{\linespread{1}\selectfont{\small The order between extensions and regularizations matters.
   }}
  \label{orderex}
\end{figure}

In~\cite{Tomamichel2015} it was shown that if $\xD$ is the classical R\'enyi entropy with $\alpha\in[1/2,\infty]$ then $\underline{\xD}^{\reg}$ is the sandwiched quantum relative entropy of order $\alpha$~\cite{WWY2014,MDS+2013}. Therefore, in this case, $\underline{\xD}_{ch}(\mM\|\mN)=D_\alpha(\mM\|\mN)$ is simply the channel extension of the sandwich relative entropy so that $\underline{\xD}_{ch}^{\reg}(\mM\|\mN)=D_\alpha^{\reg}(\mM\|\mN)$ is just the regularization of $D_\alpha(\mM\|\mN)$. In Theorem~\ref{eqreg} above we showed that for $\alpha=1$, 
$
D^{\reg}_\alpha(\mM\|\mN)=\uD^{\reg}_{\alpha}(\mM\|\mN)
$
which means that for $\alpha=1$, $D_\alpha^{\reg}$ is the smallest weakly additive divergence that reduces to the KL relative  entropy on classical states. The question remains open if this equality holds for all $\alpha\in[1/2,\infty]$.

\subsubsection{The maximal channel extension}

To the author's knowledge, the maximal channel extension in~\eqref{123g} is new, and was not studied before.
Note that the infimum in~\eqref{123g} can be expressed as 
\be
\overline{\xD}(\mM\|\mN)\eqdef\inf_{\substack{|X|\in\mbb{N}\;,\;\p,\q\in\md(X)\\ \mM=\sum_xp_x\mE_x\;,\;\mN=\sum_xq_x\mE_x\\\{\mE_x\}\subset\cptp(A\to B)}}\xD(\p\|\q)
\ee
For the case that $\mM$ is an isometry we get the following result.
\begin{theorem}\label{lem.3}
Let $\mV\in\cptp(A\to B)$ be an isometry channel defined via $\mV(\rho)=V\rho V^*$, for all $\rho\in\md(A)$, and with isometry matrix $V$ (i.e. $V^*V=I^A$). Then, for any $\mN\in\cptp(A\to B)$
\be
\overline{\xD}(\mV\|\mN)=D_{\max}(\mV\|\mN)=\log\tr\left[ J_\mN^{-1}J_\mV\right]\;.
\ee
\end{theorem}

In section~\ref{classical} we  provided a closed formula of this divergence for the classical case. We saw that the formula reveals that this divergence is not additive (even for classical channels), and therefore is not a relative entropy. Recall, however, that from Theorem~\ref{outs} we know that if $\xD$ is a classical relative entropy then its channel extension $\overline{\xD}$ is sub-additive. This means that the limit in its regularization exists, and $\overline{\xD}^{\reg}\leq\overline{\xD}$. The divergence $\overline{\xD}^{\reg}$ is weakly additive, and it remains open to determine if it is a relative entropy (i.e. fully additive).

\subsubsection{The geometric channel relative entropy}

Given a classical divergence $\xD$, its maximal extension to quantum states is given for all $\rho,\sigma\in\md(A)$ by 
\be
\overline{\xD}(\rho\|\sigma)\eqdef\inf_{\substack{|X|\in\mbb{N},\;\p,\q\in\md(X) \\
\mE(\p)=\rho,\;\mE(\q)=\sigma\\ \mE\in\cptp(X\to A)}}\xD(\p\|\q)\;.
\ee
The geometric divergence is defined as the minimal channel-extension of this maximal state-extension of $\xD$.
It can be expressed as (cf.~\eqref{asd})
\ba
\widehat{\xD}(\mM\|\mN)\eqdef\sup_{\psi\in\md(R A)} \overline{\xD}\Big(\mM^{A\to B}\left(\psi^{R A}\right)\Big\|\mN^{A\to B}\left(\psi^{R A}\right)\Big)
\ea
For the case that $\xD=D_\alpha$ is the classical R\'enyi entropy with $\alpha\in(0,2]$ and $\alpha\neq 1$, it was proved in~\cite{Fang2019,Katariya2020} that
\be\label{formula}
\widehat{\xD}(\mM\|\mN)=\widehat{D}_\alpha(\mM\|\mN)\eqdef\frac{1}{\alpha-1}\log\widehat{Q}_\alpha(\mM\|\mN)
\ee
where 
\begin{widetext}
\be
\widehat{Q}_\alpha(\mM\|\mN)\eqdef\begin{cases}\left\|\tr_B\left[G_\alpha(J_{\mM}^{AB},J_{\mN}^{AB})\right]\right\|_{\infty} & \text{if }\alpha\in(1,2]\text{ and }\supp(J_\mM^{AB})\subseteq\supp(J_\mN^{AB})\\
\lambda_{\min}\left(\tr_B\left[G_\alpha(J_{\mM}^{AB},J_{\mN}^{AB})\right]\right) & \text{if }\alpha\in(0,1)\text{ and }\supp(J_\mM^{AB})\subseteq\supp(J_\mN^{AB})\\
\infty & \text{if }\alpha\in(1,2]\text{ and }\supp(J_\mM^{AB})\not\subseteq\supp(J_\mN^{AB})\\
\lim\limits_{\epsilon\to 0^+}\lambda_{\min}\left(\tr_B\left[G_\alpha(J_{\mM_\epsilon}^{AB},J_{\mN}^{AB})\right]\right) &\text{if }\alpha\in(0,1)\text{ and }\supp(J_\mM^{AB})\not\subseteq\supp(J_\mN^{AB})
\end{cases}
\ee
\end{widetext}
with $J_{\mM_\epsilon}^{AB}\eqdef J_{\mM}^{AB}+\epsilon I^{AB}$ and 
\be
G_\alpha(X,Y)\eqdef Y^{\frac12}\left(Y^{-\frac12}X Y^{-\frac12}\right)^\alpha Y^{\frac12}\quad\forall\;X,Y> 0
\ee
For $\alpha=1$ it is given by
\begin{align*}
&\widehat{D}_{\alpha=1}(\mM\|\mN)=\widehat{D}(\mM\|\mN)\eqdef\\
&\begin{cases}
\left\|\tr_B\left[\widehat{G}(J_{\mM}^{AB},J_{\mN}^{AB})\right]\right\|_{\infty} &\text{if }\supp(J_\mM^{AB})\subseteq\supp(J_\mN^{AB})\\
\infty &\text{otherwise}
\end{cases}
\end{align*}
where 
\be
\widehat{G}(X,Y)\eqdef X^{\frac12}\log\left(X^{\frac12}Y^{-1}X^{\frac12}\right)X^{\frac12}\;.
\ee

For all $\alpha\in(0,2]$, the formula above gives for an isometry $\mV$ and a channel $\mN$
\be
\widehat{D}_\alpha(\mV\|\mN)=D_{\max}(\mV\|\mN)\;.
\ee
Hence, due to Theorem~\ref{lem.3}, $\widehat{D}_\alpha(\mV\|\mN)$ coincides with the maximal quantum-channel extension $\bD_\alpha(\mV\|\mN)$. However, recall that in general we have $\widehat{D}_\alpha(\mM\|\mN)\leq \bD_\alpha(\mM\|\mN)$, and we saw in Sec.~\ref{classical} that the inequality can be strict even on classical channels. In fact, since $\widehat{D}_\alpha$ is additive, we must have  $\widehat{D}_\alpha(\mM\|\mN)\leq \bD_\alpha^\reg(\mM\|\mN)$, and it is left open if this inequality can be strict for some choices of $\mM$ and $\mN$.

The formula~\eqref{formula} reveals that $\widehat{D}_\alpha$ is additive (at least for $\alpha\in(0,2]$).
In fact, to the authors knowledge, with the exception of $D_{\max}$, this function is the only known channel relative entropy, since all other channel divergences discussed in this paper are at most known to be weakly additive and the question whether they are fully additive is open.  In~\cite{Fang2019} $\widehat{D}_\alpha$ was used to derive upper bounds on certain QIP tasks, and in~\cite{Katariya2020} it was used to upper-bound some optimal rates in the context of channel discrimination.

 \section{Outlook}

\subsection{Some applications of channel divergences}

One of the fundamental tasks in quantum information theory is the distinguishability of two quantum channels $\mN,\mM\in\cptp(A\to B)$. In~\cite{Wang2019} it was shown that the asymptotically optimal discrimination rate of 
the type-II error exponent for parallel strategy is given by $D^\reg(\mN\|\mM)$, and for the adaptive strategy by the amortized divergence as given in~\eqref{amortized}. Remarkable, in~\cite{Fang2020} it was shown that the amortized divergence equals $D^\reg$ indicating that both the parallel and adaptive strategies yield the same error exponent.

Despite the physical significance of $D^\reg(\mN\|\mM)$ for channel discrimination, we know very little about this quantity due to the fact that it is defined in terms of an optimization over a non-compact domain. Therefore, in such scenarios, it is often useful to study the properties of such uncomputable functions. 

It is known~\cite{Fang2020b} that $D^\reg(\mN\|\mM)$ is an additive function and therefore it fits our definition of a relative entropy. This immediately implies that it satisfies all the properties outlined in Theorem~\ref{properties}. Particularly, it satisfies the continuity properties given in~\eqref{conti} and~\eqref{triang} indicating that if two channels $\mN,\mN'\in\cptp(A\to B)$ are close (in terms of the Thompson metric $D_T$), then they have similar discrimination rates with any other channel $\mM\in\cptp(A\to B)$. 
We point out that even though the Umegaki relative entropy is continuous (in fact asymptotically continuous) one cannot conclude that these properties carry over to $D^\reg$ due to the optimization involved with unbounded dimensions.

Theorem~\ref{eqreg} establishes that $D^\reg(\mN\|\mM)$ is the smallest channel relative entropy among all channel relative entropies that reduces to the KL-divergence on classical states. This means that it is not possible to find a computable lower bound  for $D^\reg(\mN\|\mM)$ that is itself a channel divergence. On the other hand, the geometric channel relative entropy $\widehat{D}$ provides such an upper bound~\cite{Fang2019}. 

Channel relative entropies have several other applications in quantum Shannon theory. For example, consider two distant parties, Alice and Bob, who share a classical-quantum channel $\mN\in\cptp(X\to B)$. This channel can be used to transmit classical information at a rate that is given by the Holevo information (see e.g.~\cite{Wilde2013}). Interestingly,  the Holevo information $\chi(\mN)$ can be expressed as
\be
\chi\left(\mN^{X\to B}\right)=\min_{\mE}D\left(\mN^{X\to B}\big\|\mE^{X\to B}\right)
\ee
where the minimum is over all replacement channels of the form $\mE(|x\lr x|)=\sigma^B$ for all $x$, and a fixed $\sigma\in\md(B)$. We can therefore interpret the Holevo quantity as a channel divergence distance to the set of ``free" channels.
In this context, the only free channels are the replacement channels. Also the entanglement-assisted classical communication rate can be expressed in a similar fashion. These examples indicate that also in the channel domain divergences and relative entropies as defined in this paper are expected to play a key role in the recently active developing field of dynamical quantum resource theories~\cite{Cooney-2016a,Felix2018,Berta2018,Gour2019,Fang2019,Kaur-2018a,Gour2018b,LW2019,LY2020,GW2019,Wang2019,Wang2019b,GS2020,Bauml2020,Fang2019,Fang2020,Katariya2020}. 

\subsection{Conclusions}

In this paper we introduced an axiomatic approach to dynamical divergences. This approached is minimalistic in the sense that  we only require channel divergences to satisfy the generalized DPI under superchannels, and channel relative entropies to be in addition additive and normalized. Remarkably, we showed that these axioms are sufficient to induce enough structure, leading to numerous properties satisfied by all channel relative entropies. One of our main results, is a uniqueness theorem, Theorem~\ref{mainresult}, in which we show that in the classical domain, there exists only one channel relative entropy, that reduces to the Kullback-Leibler divergence on classical states (i.e. probability vectors). In the quantum case, it is known that this uniqueness does not hold even for quantum states, but we were able to show that the amortized relative entropy as defined in~\eqref{amortized} is the smallest channel relative entropy that reduces to the Kullback-Leibler divergence on classical states.
Due to the one-to-one correspondence between classical entropies and classical relative entropies~\cite{GT2020a}, this means that the amortized relative entropy is in fact the smallest one that reduces on a pair of classical states $(\p,\u)$ (here $\u$ is the uniform distribution) to the log of the dimension minus the Shannon entropy.

There are many open problems for future investigations. For example, in the classical domain, for $\alpha\neq 1$, is the classical-channel extension of the $\alpha$-R\'enyi relative entropy unique? Another interesting problem is whether the regularization of the maximal channel extension of relative entropies coincide with the geometric channel relative entropies. Finally, another interesting question is whether the maximal channel extensions of relative entropies satisfy a `chain rule' similar to the one satisfied by the minimal channel extension~\eqref{qdsup}~\cite{Fang2020}. 

\begin{acknowledgments}
The author would like to thank Marco Tomamichel for numerous fruitful discussions on the subject of this paper. The author is also grateful for many fruitful discussions with Nilanjana Datta, Kun Fang, Xin Wang, and Mark Wilde on topics related to this work.
The author acknowledge support from the Natural Sciences and Engineering Research Council of Canada (NSERC).
\end{acknowledgments}

\bibliography{QRTbib}

\newpage
\onecolumngrid



\section{Supplemental Material}

\begin{center}\Large\bfseries
Proof of Theorem~\ref{properties}
\end{center}

We prove all the properties outlined in Theorem~\ref{properties} as separated lemmas.

\begin{lemma}
Channel relative entropies are normalized divergences.
\end{lemma}
\begin{proof}
Let $\sD$ be a channel relative entropy, and $\mN,\mM\in\cptp(A\to B)$. Then, $\sD(\mN\|\mM)=\sD(\mN\otimes 1\|\mM\otimes 1)=\sD(\mN\|\mM)+\sD(1\|1)$. Hence, $\sD(1\|1)=0$.
\end{proof}
\begin{lemma}
Let $\sD$ be a channel divergence. Then, for any two channels $\mM,\mN\in\cptp(A\to B)$ we have
\be
\sD(\mM\|\mN)\geq 0\;,
\ee
with equality if $\mM=\mN$.
\end{lemma}
\begin{proof}
Let $\Theta\in\super(AB\to X)$ be a superchannel that takes quantum channels to classical states. Then,
\be
\sD(\mM\|\mN)\geq\sD\left(\Theta[\mM]\big\|\Theta[\mN]\right)\geq 0
\ee 
with equality if $\Theta[\mN]=\Theta[\mM]$, where we used the fact that on classical states $\sD$ is a classical divergence.
\end{proof}

\begin{lemma}
Let $\mE_{\sigma_1},\mE_{\sigma_2}\in\cptp(A\to B)$ be two replacement channels as defined in~\eqref{replace}, with $\sigma_1,\sigma_2\in\md(B)$ and $|A|>1$. Let $\sD$ be a channel divergence. Then,
\be
\sD\left(\mE_{\sigma_1}\big\|\mE_{\sigma_2}\right)=\sD(\sigma_1\|\sigma_2)\;,
\ee  
where on the RHS, $\sigma_1$ and $\sigma_2$ are viewed as channels in $\cptp(1\to B)$.
\end{lemma}

\begin{proof}
The proof is a simple consequence of the DPI. On one direction, we define a superchannel $\Theta\in\super(AB\to B)$ by imputing the state $|0\lr 0|\in\md(A)$ into the input of a channel; specifically, for any $\mE\in\cptp(A\to B)$
\be
\Theta[\mE]\eqdef\mE(|0\lr 0|)\;.
\ee
With this superchannel we get
\ba
\sD\left(\mE_{\sigma_1}\big\|\mE_{\sigma_2}\right)&\geq \sD\left(\Theta\left[\mE_{\sigma_1}\right]\big\|\Theta\left[\mE_{\sigma_2}\right]\right)\\
&=\sD\left(\sigma_1\big\|\sigma_2\right)\;.
\ea
For the other direction, we define a superchannel $\Upsilon\in\super(B\to AB)$ as follows. For any ``channel" $\sigma\in\cptp(1\to B)$ we define
\be
\Upsilon\left[\sigma^{1\to B}\right]=\sigma^{1\to B}\circ\tr\;,
\ee
where the trace is acting on the input system $A$.
With this superchannel we get
\ba
\sD\left(\sigma_1\big\|\sigma_2\right)&\geq \sD\left(\Upsilon[\sigma_1]\big\|\Upsilon[\sigma_2]\right)\\
&=\sD\left(\sigma^{1\to B}_1\circ\tr\;\big\|\;\sigma^{1\to B}_2\circ\tr\right)\\
&=\sD\left(\mE_{\sigma_1}\big\|\mE_{\sigma_2}\right)\;.
\ea
This completes the proof.
\end{proof}

The next property is an extension of Lemma~5 in~\cite{GT2020a} to channels. 
\begin{lemma}\label{norma1}
Let $\sD$ be a channel relative entropy, and let $\mE_1,...,\mE_n\in\cptp(A\to B)$ be a set of $n$ orthogonal quantum channels; i.e. their Choi matrices satisfies $\tr[J_{\mE_j}^{AB}J_{\mE_k}^{AB}]=0$ for all $j\neq k\in[n]$. Then, for any probability vector $\p=\{p_x\}_{x=1}^n$, and $\mN\eqdef\sum_{x=1}^np_x\mE_x$, we have 
\be
\sD\left(\mE_x\big\|\mN\right)=-\log(p_x)\quad\forall\;x=1,...,n.
\ee
\end{lemma}
\begin{proof}
Since $\mE_1,...,\mE_n$ are orthogonal there exists a process POVM $\Theta\in\super(AB\to X)$ such that
\be
\Theta[\mE_x]=|x\lr x|\quad\quad\forall\;x=1,...,n.
\ee
We therefore get from the DPI
\ba
\sD\left(\mE_x\big\|\mN\right)&\geq \sD\left(\Theta[\mE_x]\big\|\Theta[\mN]\right)\\
&=\sD\Big(|x\lr x|\Big\|\sum_{y=1}^np_y|y\lr y|\Big)\\
&=-\log(p_x)\;,
\ea
where the last line follows from the classical version of this lemma that was proven in~\cite{GT2020a}.
To prove the other side of the inequality, let $X$ be an $n$-dimensional classical system, and let $\Upsilon\in\super(X\to AB)$ be a superchannel defined via
\be
\Upsilon\big[|x\lr x|\big]\eqdef\mE_x\;.
\ee
Applying again the DPI gives
\ba
-\log(p_x)&=\sD\Big(|x\lr x|\Big\|\sum_{y=1}^n p_y|y\lr y|\Big)\\
&\geq \sD\Big(\Upsilon\big[|x\lr x|\big]\Big\|\Upsilon\Big[\sum_{y=1}^np_y|y\lr y|\Big]\Big)\\
&=\sD\left(\mE_x\big\|\mN\right).
\ea
This completes the proof.
\end{proof}

Note that in the proof above, to obtain the bound $\sD\left(\mE_x\big\|\mN\right)\leq-\log(p_x)$ we did not assume that that the channels $\{\mE_x\}$ are orthogonal. Therefore, this bounds also holds if $\{E_x\}$ are not orthogonal. However, we now prove a stronger version of this bound.

\begin{definition}
Let $\sD$ be a channel divergence.
For any $\mM,\mN\in\cptp(A\to B)$ we denote its corresponding min and max divergences, respectively, by
\begin{align*}
&\sD_{\min}\left(\mM\|\mN\right)\eqdef \sD\left(\begin{bmatrix} 1 & 0\\ 0 & 0\end{bmatrix}\Big\| \bb 2^{-D_{\min}(\mM\|\mN)} & 0\\ 0 & 1-2^{-D_{\min}(\mM\|\mN)}\eb\right)\\
&\sD_{\max}\left(\mM\|\mN\right)\eqdef \sD\left(\begin{bmatrix} 1 & 0\\ 0 & 0\end{bmatrix}\Big\| \bb 2^{-D_{\max}(\mM\|\mN)} & 0\\0 & 1-2^{-D_{\max}(\mM\|\mN)}\eb\right)\;.
\end{align*}
where the classical states in the input of $\sD$ on the RHS can be viewed as channels with trivial input systems.
\end{definition}
\begin{remark}
Note that if $\sD$ is a relative entropy then from Lemma~\ref{norma1} (see also~Lemma~5 in~\cite{GT2020a}) it follows that $\sD_{\max}=D_{\max}$ and $\sD_{\min}=D_{\min}$. 
\end{remark}

\begin{lemma}\label{thm:minmax}
Let $\sD$ be a channel divergence.
Then, $\sD_{\max}$ and $\sD_{\min}$ are also divergences, and furthermore,
\be
\sD_{\min}\left(\mM\|\mN\right)\leq \sD\left(\mM\|\mN\right)\leq \sD_{\max}\left(\mM\|\mN\right)\;.
\ee
In particular, if $\sD$ is a relative entropy then
for all $\mM,\mN\in\cptp(A\to B)$
\be
D_{\min}(\mM\|\mN)\leq \sD(\mM\|\mN)\leq D_{\max}(\mM\|\mN)\;,
\ee
where $D_{\min}$ and $D_{\max}$ are the channel min and max relative entropies as defined in~\eqref{min2} and~\eqref{max1}.
\end{lemma}

\begin{proof}
We start by proving that $\sD_{\max}$ and $\sD_{\min}$ satisfy the DPI.
For this purpose, observe first that for any two binary probability distributions $(p,1-p)$ and $(q,1-q)$
there exists a classical channel $\mC\in\cptp(X\to X)$ satisfying
\be\label{cc}
\mC(|0\lr 0|)=|0\lr 0|\quad\text{and}\quad\mC\left(p|0\lr 0|+(1-p)|1\lr 1|\right)=q|0\lr 0|+(1-q)|1\lr 1|
\ee
if and only if $p\leq q$. Now, by definition, for any superchannel $\Theta\in\super(AB\to A'B')$
\ba
&\sD_{\max}\left(\Theta[\mM]\big\|\Theta[\mN]\right)= \sD\left(|0\lr 0|\;\Big\|\; 2^{-D_{\max}(\Theta[\mM]\|\Theta[\mN])}|0\lr 0|+\left(1-2^{-D_{\max}(\Theta[\mM]\|\Theta[\mN])}\right)|1\lr 1|\right)
\ea
and also $2^{-D_{\max}(\Theta[\mM]\|\Theta[\mN])}\geq 2^{-D_{\max}(\mM\|\mN)}$. This means that there exists a classical channel $\mC\in\cptp(X\to X)$ satisfying~\eqref{cc} with $q=2^{-D_{\max}(\Theta[\mM]\|\Theta[\mN])}$ and $p=2^{-D_{\max}(\mM\|\mN)}$. Hence, with this classical channel $\mC$ we get
\ba
\sD_{\max}\left(\Theta[\mM]\big\|\Theta[\mN]\right)&= \sD\Big(\mC(|0\lr 0|)\;\Big\|\; \mC\left(2^{-D_{\max}(\mM\|\mN)}|0\lr 0|+\big(1-2^{-D_{\max}(\mM\|\mN)}\big)|1\lr 1|\right)\Big)\\
&\leq \sD\Big(|0\lr 0|\;\Big\|\; 2^{-D_{\max}(\mM\|\mN)}|0\lr 0|+\left(1-2^{-D_{\max}(\mM\|\mN)}\right)|1\lr 1|\Big)\\
&=\sD_{\max}\left(\mM\|\mN\right)\;.
\ea
Following the same lines as above, one can prove that also $\sD_{\min}$ satisfies the DPI. We are now ready to prove the two bounds.

Let $\mM\in\cptp(A\to B)$, $\psi\in\md(RA)$, and $\Pi_{\mM(\psi)}$ denotes the projector to the support of $\mM^{A\to B}(\psi^{RA})$. Define the superchannel (in fact process POVM) $\Theta\in\super(AB\to X)$ with $|X|=2$ as
\be
\Theta[\mN]\eqdef\tr\big[\mN(\psi)\Pi_\mM(\psi)\big]|0\lr 0|^X+\tr\Big[\mN(\psi)\left(I-\Pi_\mM(\psi)\right)\Big]|1\lr 1|^X\;.
\ee
Then,
\ba\label{substi}
\sD(\mM\|\mN)&\geq \sD\big(\Theta[\mM]\big\|\Theta[\mN]\big)\\
&=\sD\left(|0\lr 0|\Big\|\tr\big[\mN(\psi)\Pi_\mM(\psi)\big]|0\lr 0|+\tr\Big[\mN(\psi)\left(I-\Pi_\mM(\psi)\right)\Big]|1\lr 1|\right)\;.
\ea
Now, since the above equation holds for all $\psi^{RA}$ it also holds for the optimal $\psi^{RA}$ that satisfies 
\be
\tr\big[\mN(\psi)\Pi_\mM(\psi)\big]=2^{-D_{\min}(\mM\|\mN)}\;.
\ee
Therefore, substituting this choice of $\psi$ in~\eqref{substi} gives $\sD(\mM\|\mN)\geq \sD_{\min}(\mM\|\mN)$.

For the second inequality, denote by $t=2^{D_{\max}(\mM\|\mN)}$, and note that in particular, $t\mN\geq\mM$ (i.e. $t\mN-\mM$ is a CP map). Define a superchannel $\Theta\in\cptp(X\to AB)$ with $|X|=2$ by 
\be
\Theta\left[|0\lr 0|\right]=\mM\quad\text{and}\quad\Theta\left[|1\lr 1|\right]=\frac{1}{t-1}(t\mN-\mM)\;.
\ee
Furthermore, denote 
\be
\q^X\eqdef\frac{1}{t}|0\lr 0|^X+\frac{t-1}{t}|1\lr 1|^X\;,
\ee
 and observe that $\Theta\left[\q^X\right]=\mN$.
Hence,
\ba
\sD(\mM\|\mN)&=\sD\Big(\Theta\left[|0\lr 0|^X\right]\big\|\Theta[\q^X]\Big)\\
&\leq \sD\left(|0\lr 0|^X\big\|\q^X \right)\\
&=\sD_{\max}(\mM\|\mN)\;.
\ea
This completes the proof.
\end{proof}

Finally, the last basic property of channel divergences that we study in this subsection is the faithfulness of a channel divergence.
Similar to the classical and quantum cases, a channel divergence, $\sD$, is said to be faithful if $\sD(\mM\|\mN)=0$ implies $\mM=\mN$.
We now prove that channel divergences are faithful if their reduction to classical states is faithful.

\begin{lemma}
Let $\sD$ be a channel divergence. Then, $\sD$ is faithful if and only if its reduction to classical (diagonal) states is faithful.
\end{lemma}
\begin{proof}
Clearly, if $\sD$ is faithful on quantum channels it is also faithful on classical channels as the latter is a subset of the former. Suppose now that $\sD$ is faithful on classical states, and suppose by contradiction that there exists $\mM\neq\mN\in\md(A)$ such that $\sD(\mM\|\mN)=0$. Then, there exists a state $\psi\in\md(RA)$ and  basis of $RB$ such that the diagonal of $\mM(\psi^{RA})$ in this basis does not equal to the diagonal of $\mN(\psi^{RA})$. Let $\Delta\in\cptp(RB\to Z)$ be the completely dephasing channel in this basis, where $Z\cong RB$ is viewed as a classical system with respect to this basis. We therefore have $\Delta\circ\mM(\psi^{RA})\neq\Delta\circ\mN(\psi^{RA})$. 
Now, define a superchannel $\Theta\in\super(AB\to Z)$ via
\be
\Theta[\mE^{A\to B}]\eqdef\Delta^{RB\to Z}\circ\mE^{A\to B}\circ\psi^{1\to RA}\;,
\ee
where we view the state $\psi$ as a channel in $\cptp(1\to RA)$. We therefore get that
\be
0=\sD\big(\mM\big\|\mN\big)\geq\sD\left(\Theta[\mM]\big\|\Theta[\mN]\right)\;.
\ee
But since $\sD$ is faithful on diagonal states we must have $\Theta[\mM]=\Theta[\mN]$ in contradiction with $\Delta\circ\mM(\psi^{RA})\neq\Delta\circ\mN(\psi^{RA})$.
Hence, $\sD$ is faithful also on quantum channels.
\end{proof} 

Note that the lemma above implies that any channel divergence, that reduces to the R\'enyi relative entropy of order $\alpha>0$, is faithful. More generally, when combining the above lemma with the condition on faithfulness in Theorem~17 of~\cite{GT2020a}  we get that almost all channel relative entropies are faithful.

\begin{lemma}\label{nord}
Let $\sD$ be a channel relative entropy, and consider two systems $A$ and $B$ with $|A|\leq |B|$. Let $\mR\in\cptp(A\to B)$ be the completely randomizing channel, and $\mV\in\cptp({A\to B})$ be an isometry channel. Then,
\be
\sD\left(\mV^{A\to B}\big\|\mR^{A\to B}\right)=\log|AB|\;.
\ee
\end{lemma}
\begin{proof}
From Lemma~\ref{thm:minmax} we have
\ba
\sD\left(\mV^{A\to B}\big\|\mR^{A\to B}\right)&\leq D_{\max}\left(\mV^{A\to B}\big\|\mR^{A\to B}\right)\\
&=\log\min\{t:\;t\u^A\otimes\u^B\geq \phi^{AB}_+\}=\log|AB|\;.
\ea
On the other hand, 
\ba
\sD\left(\mV^{A\to B}\big\|\mR^{A\to B}\right)&\geq \max_{\psi\in\md(RA)} \sD\Big(\mV^{A\to B}(\psi^{RA})\big\|\mR^{A\to B}(\psi^{RA})\Big)\\
&=\max_{\psi\in\md(RA)} \sD\Big(\mV^{A\to B}(\psi^{RA})\big\|\psi^{R}\otimes\u^B)\Big)\\
&\geq \sD\Big(\mV^{A\to B}(\phi^{RA}_+)\big\|\u^{R}\otimes\u^B)\Big)\\
&=\log|AB|\;,
\ea
where the first inequality follows from~\eqref{ncad}, and the last equality follows from Lemma 5 of~\cite{GT2020a}.
\end{proof}

\begin{lemma}
Let $\sD$ be a channel relative entropy. Then, for any $\mN,\mM,\mE\in\cptp(A\to B)$ we have
\be
\sD(\mN\|\mM)\leq\sD(\mN\|\mE)+D_{\max}(\mE\|\mM)
\ee
\end{lemma}
\begin{proof}
For any $\epsilon>0$, consider the relation
	\begin{align}
		\mM = (1-\eps) \mE + \eps \mF , \quad \textrm{where} \quad \mF \eqdef \frac1\eps\left(\mM - (1-\eps) \mE\right)\;.
	\end{align}
	Note that $\mF\in\cptp(A\to B)$ if and only if $\mM - (1-\eps) \mE\in\cp(A\to B)$ or equivalently if and only if 
\be
	\eps\geq 1-2^{-D_{\max}
	(\mE\|\mM)}\;.
\ee
Set $\eps\eqdef 1-2^{-D_{\max}(\mE\|\mM)}$, and define the superchannel $\Theta\in\super(AB\otimes X\to AB)$, where $X$ is a 2-dimensional classical system, such that for all $\mR\in\cptp(A\to B)$
\be
\Theta\left[\mR^{A\to B}\otimes |0\lr 0|^X\right]\eqdef \mR^{A\to B}\quad\text{and}\quad\Theta\left[\mR^{A\to B}\otimes |1\lr 1|^X\right]\eqdef \mF^{A\to B}\;.
\ee
Such a superchannel can be realized by performing a basis measurement on $X$. If the outcome is 0 then nothing is done to the dynamical system $AB$, whereas if outcome 1 occurs then the input channel of the dynamical system $AB$ is replaced with the channel $\mF^{A\to B}$. Using additivity and the fact that $\sD\left(|0\lr 0|^X\big\|t|0\lr 0|^X+(1-t)|1\lr 1|^X\right)=-\log t$ for any $t\in[0,1]$ (see Lemma 5 of~\cite{GT2020a}) we get
\ba
\sD(\mN\|\mE)&=\sD\left(\mN\otimes|0\lr 0|^X\big\|\mE\otimes((1-\eps)|0\lr 0|^X+\eps|1\lr 1|^X)\right)+\log(1-\eps)\\
&\geq \sD\left(\Theta\left[\mN\otimes|0\lr 0|^X\right]\Big\|\Theta\left[\mM\otimes\left((1-\eps)|0\lr 0|^X+\eps|1\lr 1|^X\right)\right]\right)+\log(1-\eps)\\
&=\sD\left(\mN\big\|(1-\eps)\mE+\eps\mF)\right)+\log(1-\eps)\\
&=\sD\left(\mN\big\|\mM)\right)-D_{\max}(\mE\|\mM)\;,
\ea 
where the inequality follows from the DPI.
\end{proof}

\begin{lemma}
	Let $\mN, \mE, \mM \in\cptp(A\to B)$ be quantum channels. Then, 
	\begin{align}
		\sD(\mN\|\mM) - \sD(\mE\|\mM) &\leq \min_{s\leq 2^{-D_{\max}(\mE\|\mN)}} D_{\max}\left(\mN+s(\mM-\mE)\big\|\mM\right)\;.
		\end{align}
		Moreover, if $J_{\mN},\;J_{\mE},$ and $J_{\mM}$ have full support then
		\begin{align}\label{bound}
	\sD(J_{\mN}\|J_{\mM}) - \sD(J_{\mE}\|J_{\mM})\leq \log \left( 1 + \frac{  \| J_{\mN} - J_{\mE} \|_{\infty} } {\lambda_{\min}(J_{\mE}) \lambda_{\min}(J_{\mM}) }  \right)\;.
	\end{align}
\end{lemma}

\begin{proof}
Let $\Theta\in\super(AB\to AB)$ be a superchannel with $|X|=2$ define for any $\mP\in\ml(A\to B)$ by
		\begin{align}
		\Theta[\mP]\eqdef (1-\eps) \mP + \eps \mF \quad \textrm{where} \quad \mF \eqdef \frac1\eps\left(\mN - (1-\eps) \mE\right)\;. 
	\end{align}
	Again, the condition $\epsilon\geq1-2^{-D_{\max}
	(\mE\|\mN)}$ is equivalent to $\mF\geq 0$ which ensures that $\Theta$ is indeed a superchannel.
	By definition, $\Theta[\mE] = \mN$. Let $t>1$ be the smallest number satisfying
	\be
	\mR\eqdef\frac{1}{t-1}\left(t\mM-\Theta[\mM]\right)\geq 0\;.	 
	\ee
	That is, $t=2^{D_{\max}\left(\Theta[\mM]\|\mM\right)}$ and $\mR\in\cptp(A\to B)$. Finally, define the superchannel 
	$\Upsilon\in\super(AB\otimes X\to AB)$ such that for any $\mP\in\cptp(A\to B)$ 
	\be
	\Upsilon\left[\mP\otimes |x\lr x|^X\right]\eqdef\begin{cases}
	\Theta[\mP] &\text{if }x=0\\
	\mR &\text{if }x=1
	\end{cases}\;.
	\ee
	Using the DPI with this superchannel gives
	\ba
	\sD(\mE\|\mM)+\log t&=\sD\Big(\mE\otimes |0\lr 0|\Big\|\mM\otimes (t^{-1}|0\lr 0|+(1-t^{-1})|1\lr 1|)\Big)\\
	&\geq \sD\Big(\Upsilon\left[\mE\otimes |0\lr 0|\right]\Big\|\Upsilon\left[\mM\otimes (t^{-1}|0\lr 0|+(1-t^{-1})|1\lr 1|)\right]\Big)\\
	&=\sD\Big(\Theta\left[\mE\right]\Big\|t^{-1}\Theta\left[\mM\right]+(1-t^{-1})\mR\Big)\\
	&=\sD\big(\mN\big\|\mM\big)\;.
	\ea
	Hence,
	\ba
	\sD\big(\mN\big\|\mM\big)-\sD(\mE\|\mM)&\leq\log t\\
	&=D_{\max}\left(\Theta[\mM]\|\mM\right)\\
	&=D_{\max}\left((1-\eps) \mM + \eps \mF\|\mM\right)\\
	&=D_{\max}\left(\mN+s(\mM-\mE)\big\|\mM\right)\;,
	\ea
	where $s:=1-\epsilon\leq 2^{-D_{\max}(J_{\mE}\|J_{\mN})}$. Finally, to see the bound~\eqref{bound}, observe that
	\be
	D_{\max}\left(\mN+s(\mM-\mE)\big\|\mM\right)=\log\min\big\{t\geq 0\;:\;t\mM\geq \mN+s(\mM-\mE)\big\}\;.
	\ee
	Taking 
	\be
	t=1+\frac{1-s}{\lambda_{\min}(J_\mM)}\quad\text{and}\quad s=1-\frac{\|J_{\mN}-J_{\mE}\|_{\infty}}{\lambda_{\min}(J_\mE)}
	\ee
	gives the desired bound~\eqref{bound}. Note that these choices are consistent with $t\geq 0$ and $s\leq 2^{-D_{\max}(J_{\mE}\|J_{\mN})}$.
	\end{proof}

\begin{center}\Large\bfseries 
Proof of Theorem~\ref{thm2}
\end{center}

\begin{theorem*}
Let $\sD$ be a classical channel divergence that reduces to the classical (static) divergence $\xD$ on classical states in $\md(X)\times\mD(X)$. Suppose further that $\xD$ is quasi-convex. Then, for all $\mM,\mN\in\cptp(X\to Y)$
\be
\underline{\xD}(\mM\|\mN)\leq\sD(\mM\|\mN)\leq\overline{\xD}(\mM\|\mN)
\ee
where
\ba
&\underline{\xD}(\mM\|\mN)=\max_{x\in\{1,...,|X|\}}\xD\left(\mM(|x\lr x|)\big\|\mN(|x\lr x|)\right)\\
&\overline{\xD}(\mM\|\mN)\eqdef\inf_{|Z|\in\mbb{N}}\Big\{\xD\left(\p\|\q\right)\;:\;\big(\mM(|x\lr x|),\mN(|x\lr x|)\big)\prec_r(\p,\q)\;\;\forall\;x\in[|X|]\;\;,\;\;\p,\q\in\md(Z)\Big\}\;.
\ea
Moreover, $\underline{\xD}$ is a classical channel relative entropy, and $\overline{\xD}$ is a normalized classical channel divergence.
\end{theorem*}

\begin{proof}
We apply Theorem~\ref{outs} by taking $\mr(X\to Y)$ to be the set of classical replacement channels, so that the minimal and maximal extensions to all classical channels are given by
\begin{align}
&\underline{\xD}(\mM\|\mN)\eqdef\sup_{Z}\Big\{\xD(\Theta[\mM]\|\Theta[\mN])\;:\;\Theta\in\super(XY\to Z),\;\;\Theta[\mM],\Theta[\mN]\in\md(Z)\Big\}\\
&\overline{\xD}(\mM\|\mN)\eqdef\inf_{Z}\Big\{\sD(\p\|\q)\;:\;\p,\q\in\md(Z)\;\;,\;\;\exists\Theta\in\super(Z\to XY)\;\text{ s.t. }\mM=\Theta[\p]\;\text{,}\;\mN=\Theta[\q]\Big\}\;.
\end{align}
For the minimal extension, observe that any superchannel $\Theta\in\super(XY\to Z)$ can be expressed as
\be
\Theta\left[\mN^{X\to Y}\right]=\mE^{WY\to Z}\circ\mN^{X\to Y}(\p^{WX})\quad\forall\mN\in\cptp(X\to Y)\;,
\ee
where $\mE\in\cptp(WY\to Z)$ and $\p\in\md(WX)$. Therefore, since $\xD$ satisfies the DPI it follows that
\ba
\xD(\Theta[\mM]\|\Theta[\mN])&\leq \xD\left(\mN^{X\to Y}(\p^{WX})\big\|\mN^{X\to Y}(\p^{WX})\right)\\
&\leq \max_{x\in\{1,...,|X|\}}\xD\left(\mM(|x\lr x|)\big\|\mN(|x\lr x|)\right)\;,
\ea
where the last inequality follows from the quasi-convexity of $\sD$. On the other hand, the RHS can be achieved by taking $\Theta\in\super(XY\to Z)$ such that $|W|=1$, $|Z|=|Y|$, $\mE=\id^{Y\to Z}$, and $\p^X=|x\lr x|^X$. 

For the maximal extension, observe that the conditions $\mM=\Theta[\p]$ and $\mN=\Theta[\q]$ can be expressed as
\ba
&\mM^{X\to Y}\left(|x\lr x|^X\right)=\mE^{XZ\to Y}\left(|x\lr x|^X\otimes\p^Z\right)\\
&\mN^{X\to Y}\left(|x\lr x|^X\right)=\mE^{XZ\to Y}\left(|x\lr x|^X\otimes\q^Z\right)
\ea
where $\mE\in\cptp(XZ\to Y)$ is the post-processing map associated with $\Theta$, and the equalities hold for all $x=1,...,|X|$.
This condition is equivalent to 
\be
\big(\mM(|x\lr x|),\mN(|x\lr x|)\big)\prec_r(\p,\q)\;\;\forall\;x=1,...,|X|\;.
\ee
The optimality properties of the maximal and minimal extensions follow from Theorem~\ref{outs}.
\end{proof}

\begin{center}\Large\bfseries {Proof of Theorem~\ref{thm:cc}}
\end{center}

\begin{theorem*}
Let $\D$ be a classical divergences and let $\mN,\mM\in\cptp(X\to Y)$. Using the notations above, the maximal classical channel extension $\overline{\xD}$ is given by
\be
\overline{\xD}(\mM\|\mN)=\D(\p\|\q)\;,
\ee
where $\p=\{p_z\}_{z=1}^{|Y|}$ and $\q=\{q_z\}_{z=1}^{|Y|}$ are $|Y|$-dimensional probability vectors given by
\ba
&p_1=M_{1x_1}\quad{,}\quad p_z=a_{zx_z}-a_{(z-1)x_{z-1}}\\
&q_1=N_{1x_1}\quad{,}\quad q_z=b_{zx_z}-b_{(z-1)x_{z-1}}\quad\forall\;z=2,...,|Y|
\ea
where $x_1,...,x_{|Y|}\in\{1,...,|X|\}$ are defined by induction via the relations
\be
\frac{N_{1x_1}}{M_{1x_1}}=\min_{x\in[|X|]}\frac{N_{1x}}{M_{1x}}\;,
\ee
and for any $z=2,...,|Y|$
\be\label{defk}
\frac{b_{zx_z}-b_{(z-1)x_{z-1}}}{a_{zx_z}-a_{(z-1)x_{z-1}}}=\min_{x\in\{1,...,|X|\}}\frac{b_{zx}-b_{(z-1)x_{z-1}}}{a_{zx}-a_{(z-1)x_{z-1}}}\;.
\ee
\end{theorem*}
\begin{proof}
We construct the vertices $V_1,...,V_{|Y|-1}$ of the optimal lower Lorenz curve $\mL(\p,\q)$ (we do not include the vertices $(0,0)$ and $(1,1)$).
The theorem state that the first vertex of $\mL(\p,\q)$ is
$
V_1=(a_{1x_1},b_{1x_1})
$,
where $x_1$ is chosen such that
\be
\frac{b_{1x_1}}{a_{1x_1}}=\min_{x\in[|X|]}\frac{b_{1x}}{a_{1x}}\;.
\ee
The reason for that is that among all the lines connecting $(0,0)$ to $(a_{1x},b_{1x})$ (with $x=1,...,|X|$), the line connecting $(0,0)$ to $V_1$ has the smallest slope, so that it is not strictly above any of the lower Lorenz curves $\mL(\m_x,\n_x)$. 
The next vertex of $\mL(\p,\q)$, is taken to be $V_2=(a_{2x_2},b_{2x_2})$, where $x_2$ is defined via
\be
\frac{b_{2x_2}-b_{1x_1}}{a_{2x_2}-a_{1x_1}}=\min_{x\in[|X|]}\frac{b_{2x}-b_{1x_1}}{a_{2x}-a_{1x_1}}\;.
\ee
This choice of $V_2$ ensures that the line connecting $V_1$ and $V_2$ has the smallest slope among all lines connecting $V_1$ to the second vertices of the curves $\{\mL(\m_x,\n_x)\}_{x=1}^{|X|}$ (i.e. the vertices $\{(a_{2x},b_{2x})\}_{x=1}^{|X|}$). This ensures that the line connecting $V_1$ and $V_2$ is never strictly above any of the Lorenz curves $\mL(\m_x,\n_x)$. Continuing in this way, the vertices $V_k=(a_{kx_{k}},b_{kx_{k}})$, with $x_{k}$ as defined in~\eqref{defk} has the property that the slop of the line connecting $V_{k-1}$ to $V_{k}$ is the smallest one among all lines connecting $V_{k-1}$ and any of the other of the vertices  $\{(a_{kx},b_{kx})\}_{x=1}^{|X|}$.
Again, this ensures that the line connecting $V_{k-1}$ and $V_k$ is never strictly above any of the Lorenz curves $\mL(\m_x,\n_x)$. Finally, the optimality of this choice follows from the fact that we constructed $\mL(\p,\q)$ from the same vertices of all the lower Lorenz curves $\{\mL(\m_x,\n_x)\}$. Hence, any other optimal pair of probability vectors $(\p',\q')$ that its curve $\mL(\p',\q')$ is below all the curves $\mL(\m_x,\n_x)$, is also not strictly above the vertices of $\mL(\p,\q)$, so that it must also be below the curve $\mL(\p,\q)$. This completes the proof.
\end{proof}

\begin{center}\Large\bfseries {Proof of Theorem~\ref{mainresult}}\end{center}

\begin{theorem*}
Let $\sD$ be a classical channel divergence that reduces to the Kullback–Leibler divergence, $D$, on classical states. If $\sD$ is continuous in its second argument then for all $\mN,\mM\in\cptp(X\to Y)$
\be
\sD(\mM\|\mN)=\max_{x\in\{1,...,|X|\}}D\left(\mM(|x\lr x|)\big\|\mN(|x\lr x|)\right)\;.
\ee
\end{theorem*}

We first prove the following two lemmas.

\begin{lemma}\label{usf}
Let $a\in\mbb{N}$, $\r_{1},...,\r_a\in\md(Y)$, and $\s\in\md(Z)$. Then,
\be
\max_{\p\in\md(YZ)}\Big\{H(\p)\;:\;\r_x\otimes\s\prec\p\;\;,\;\;\forall\;x\in[a]\Big\}
\leq\log|Z|+\max_{\q\in\md(Y)}\Big\{H(\q)\;:\;\r_x\prec\q\;\;,\;\;\forall\;x\in[a]\Big\}
\ee
\end{lemma}
\begin{proof}
Set $m\eqdef|Z|$, $n\eqdef|Y|$, and observe that for any $x=1,...,a$ we have $\r_x\otimes\s\succ\r_x\otimes\u^{(m)}$. Therefore,
\be\label{157}
\max_{\p\in\md(YZ)}\Big\{H(\p)\;:\;\r_x\otimes\s\prec\p\;\;,\;\;\forall\;x\in[a]\Big\}
\leq \max_{\p\in\md(YZ)}\Big\{H(\p)\;:\;\r_x\otimes\u^{(m)}\prec\p\;\;,\;\;\forall\;x\in[a]\Big\}
\ee
Now, in~\cite{FGG2013} it was shown that for a given set of probability vectors $\v_1,...,\v_n\in\md(k)$ there exists an optimal vector $\u$ that satisfies (1) $\u\succ\v_x$ for all $x=1,...,n$, and (2) for any other $\mathbf{w}$ for which $\w\succ\v_x$ must also satisfy $\mathbf{w}\succ\u$. Furthermore, the components of the vector $\u=(u_1,...,u_k)$ are given by
\be
u_x=\Omega_x-\Omega_{x-1}\quad\text{where}\quad\Omega_x\eqdef\max_{x\in[n]}\sum_{y=1}^{x}(\v_x)_y^{\downarrow}
\ee
and $\{(\v_x)_y^{\downarrow}\}_y$ are the components of $\v_x$ arranged in decreasing order. 
Applying this formula to our case, we get that the optimal $\p=\p^{YZ}=\sum_{y,z}p_{yz}|yz\lr yz|^{YZ}$ of the optimization on the RHS of~\eqref{157}  has components
\be
p_{yz}=\Omega_{yz}-\Omega_{y(z-1)}\quad;\quad\Omega_{yz}\eqdef\max_{x\in[a]}\Big\{\sum_{y'=1}^{y-1}r_{y'|x}
+\frac{z}{m}r_{y|x}\Big\}
\quad;\quad y\in[n]\;,\;z\in[m]\;.
\ee
with the convention that $\Omega_{y0}\eqdef \Omega_{(y-1)m}$ and $\Omega _{10}=0$. Now, observe that
\be
\p^{YZ}\succ\p^{Y}\otimes \u^Z
\ee
since the RHS can be obtained by applying a doubly stochastic matrix to the LHS. Finally, observe that the components of $\p^Y$ are given by
\be
p_y\eqdef\sum_{z=1}^{m}p_{yz}=\Omega_{ym}-\Omega_{(y-1)m}=\max_{x\in[a]}\sum_{y'=1}^{y}r_{y'|x}-\max_{x\in[a]}\sum_{y'=1}^{y-1}r_{y'|x}
\ee 
Therefore, the vector $\p^Y$ is the optimal probability vector that satisfies $\r_x\prec\p^{Y}$ for all $x\in[a]$.
Hence,
\ba
\max_{\p\in\md(YZ)}\Big\{H(\p)\;:\;\r_x\otimes\u^{(m)}\prec\p\;\;,\;\;\forall\;x\in[a]\Big\}&=H(\p^{YZ})\\
&\leq H(\p^{Y}\otimes\u^Z)=\log(m)+H(\p^Y)\\
&=\log(m)+\max_{\q\in\md(Y)}\Big\{H(\q)\;:\;\r_x\prec\q\;\;,\;\;\forall\;x\in[a]\Big\}\;.
\ea
This completes the proof.
\end{proof}

\begin{lemma}\label{assis}
Let $\p\in\md(Y)$ and $0<\epsilon,\delta<1$. Then, for large enough $n\in\mbb{N}$
\be\label{gasd}
\p^{\otimes n}\prec\big(\delta,\underbrace{2^{-n\left(H(\p)-\epsilon\right)},...,2^{-n\left(H(\p)-\epsilon\right)}}_{c_n\text{-times}},s_n\big)
\ee
where
\be
c_n\eqdef \left\lfloor(1-\delta)2^{n\left(H(\p)-\epsilon\right)}\right\rfloor\quad\text{and}\quad s_n\eqdef 1-\delta-c_n 2^{-n\left(H(\p)-\epsilon\right)}\;.
\ee
\end{lemma}
\begin{proof}
For any $n\in\mbb{N}$ denote by $\mt_{n,\epsilon}(\p)$ the set of all $\epsilon$-typical sequences. Then, 
\be
\p^{\otimes n} \prec \pr\left(x^n\not\in\mt_{n,\epsilon}(\p)\right)\bigoplus_{x^n\in\mbb{T}_{n,\epsilon}(\p)}p_{x^n}\;.
\ee
Recall that for any $x^n\in\mt_{n,\epsilon}(\p)$ we have $p_{x^n}\leq 2^{-n\left(H(\p)-\epsilon\right)}$. Let $n$ be large enough such that  $\pr\left(x^n\not\in\mt_{n,\epsilon}(\p)\right)<\delta$. Note that $c_n$ is the largest integer such that
\be
\delta+c_n 2^{-n\left(H(\p)-\epsilon\right)}\leq 1\;.
\ee
Therefore, since $\delta+(c_n+1) 2^{-n\left(H(\p)-\epsilon\right)}> 1$ we must have
\be
s_n\eqdef 1-\delta-c_n 2^{-n\left(H(\p)-\epsilon\right)}\leq 2^{-n\left(H(\p)-\epsilon\right)}\;.
\ee
We therefore conclude that~\eqref{gasd} holds.
\end{proof}

We are now ready to prove the theorem.

\begin{proof}[Proof of Theorem~\ref{mainresult}]
We will start by computing $\bD^{\reg}$. Denote by $a\eqdef|X|$ and $b\eqdef|Y|$, and for each $x\in[a]$ denote $\m_x\eqdef\mM(|x\lr x|)\in\md(Y)$ and $\n_x\eqdef\mN(|x\lr x|)\in\md(Y)$. For any sequence $x^k\eqdef(x_1,...,x_k)\in[a]^k$ denote
\be
\m_{x^k}\eqdef\m_{x_1}\otimes\cdots\otimes\m_{x_k}
\ee
Note that we can express $\m_{x^k}$ in terms of the type of the sequence $x^k$
\be
\m_{x^k}=\m_1^{\otimes kt_1}\otimes\cdots\otimes \m_{a}^{\otimes k t_{a}}
\ee
where $\t\eqdef(t_1,...,t_a)$ is the type of the sequence $x^k$. Similarly,
$\n_{x^k}$ can also be extpressed in terms of $\t$ as
\be
\n_{x^k}=\n_1^{\otimes kt_1}\otimes\cdots\otimes \n_{a}^{\otimes k t_{a}}\;.
\ee
Denoting by $\mt_{a,k}$ the set of all types of sequences in $[a]^k$, we get by definition
\be
\bD\left(\mM^{\otimes k}\big\|\mN^{\otimes k}\right)=\inf\Big\{D(\p\|\q)\;:\;\left(\m_1^{\otimes kt_1}\otimes\cdots\otimes \m_{a}^{\otimes k t_{a}}\;,\;\n_1^{\otimes kt_1}\otimes\cdots\otimes \n_{a}^{\otimes k t_{a}}\right)\prec_r(\p,\q)\;\;,\;\;\forall\;\t\in\mt_{a,k}\Big\}
\ee
where the infimum is over all systems $Z$ and all probability vectors  $\p,\q\in\md(Z)$. Suppose first that the channel $\mN$ is such that all the probability vectors $\n_1,...,\n_a$ have positive rational components. In this case, there exists $n\in\mbb{N}$ such that
\be
\n_x=\left(\frac{n_{1|x}}{n},...,\frac{n_{b|x}}{n}\right),\quad n_{1|x},...,n_{b|x}\in\mbb{N},\quad\text{and}\quad\sum_{y=1}^{b}n_{y|x}=n\quad\forall\;x\in[a]\;.
\ee
For any $x\in[a]$ we denote the components of $\m_x$ by $\{m_{y|x}\}_{y=1}^b$, and define
\be
\r_x\eqdef\bigoplus_{y=1}^{b}m_{y|x}\u^{(n_{y|x})}\in\md(n)\;.
\ee
With these notations we have (see~\eqref{simp})
\be
(\m_x,\n_x)\sim_r(\r_x,\u^{(n)})
\ee
so that also
\be
(\m_{x^k},\n_{x^k})\sim_r(\r_{x^k},\u^{(n^k)})
\ee
where
\be
\r_{x^k}\eqdef\r_1^{\otimes kt_1}\otimes\cdots\otimes \r_{a}^{\otimes k t_{a}}\;.
\ee
The above relations give
\ba\label{1322}
\bD\left(\mM^{\otimes k}\big\|\mN^{\otimes k}\right)&=\inf\Big\{D(\p\|\q)\;:\;(\r_{x^k},\u^{(n^k)})\prec_r(\p,\q)\;\;,\;\;\forall\;x^k\in[a]^k\Big\}\\
&\leq \min_{\p\in\md(n^k)}\Big\{D(\p\|\u^{(n^k)})\;:\;\r_1^{\otimes kt_1}\otimes\cdots\otimes \r_{a}^{\otimes k t_{a}}\prec\p\;\;,\;\;\forall\;\t\in\mbb{T}_{a,k}\Big\}\\
&=k\log(n)-\max_{\p\in\md(n^k)}\Big\{H(\p)\;:\;\r_1^{\otimes kt_1}\otimes\cdots\otimes \r_{a}^{\otimes k t_{a}}\prec\p\;\;,\;\;\forall\;\t\in\mbb{T}_{a,k}\Big\}\\
&\leq(k+\ell a)\log(n)- \max_{\p\in\md(n^{k+a\ell})}\Big\{H(\p)\;:\;\r_1^{\otimes (kt_1+\ell)}\otimes\cdots\otimes \r_{a}^{\otimes (k t_{a}+\ell)}\prec\p\;\;,\;\;\forall\;\t\in\mbb{T}_{a,k}\Big\}
\ea
where the first inequality follows by restricting $\p$ and $\q$ to be $n^k$-dimensional and taking $\q=\u^{(n^k)}$. 
In the last inequality, $\ell$ is an integer $1\leq\ell\leq k$, and the inequality follows from Lemma~\ref{usf}.

Denote by $k_x\eqdef kt_x+\ell$ and observe that from Lemma~\ref{assis} we have for any $\delta,\epsilon>0$ and large enough $\ell\in\mbb{N}$ 
\be
\r_x^{\otimes k_x}\prec\big(\delta,\underbrace{2^{-k_x\left(H(\r_x)-\epsilon\right)},...,2^{-k_x\left(H(\r_x)-\epsilon\right)}}_{c_{x}\text{-times}},s_{x}\big)
\ee
where
\be
c_{x}\eqdef \left\lfloor(1-\delta)2^{k_x\left(H(\r_x)-\epsilon\right)}\right\rfloor\quad\text{and}\quad s_{x}\eqdef 1-\delta-c_{x} 2^{-k_x\left(H(\r_x)-\epsilon\right)}\;.
\ee
Let $z\in[a]$ be such that $H(\r_z)=\max_{x\in[a]}H(\r_x)$. Then, we get that for all $x\in[a]$ we have
\be
\r_x^{\otimes k_x}\prec\v_x\eqdef\big(\delta,\underbrace{2^{-k_x\left(H(\r_z)-\epsilon\right)},...,2^{-k_x\left(H(\r_z)-\epsilon\right)}}_{c_{x}'\text{-times}},s_{x}'\big)
\ee
where
\be
c_{x}'\eqdef \left\lfloor(1-\delta)2^{k_x\left(H(\r_z)-\epsilon\right)}\right\rfloor\quad\text{and}\quad s_{x}'\eqdef 1-\delta-c_{x}' 2^{-k_x\left(H(\r_z)-\epsilon\right)}\;.
\ee
Hence, from~\eqref{1322} and the fact that $\r_1^{\otimes k_1}\otimes \cdots\otimes\r_1^{\otimes k_1}\prec\v_1\otimes\cdots\otimes\v_a$ we get
\ba\label{11111}
\bD\left(\mM^{\otimes k}\big\|\mN^{\otimes k}\right)&\leq(k+\ell a)\log(n)- \max_{\p\in\md(n^{k+a\ell})}\Big\{H(\p)\;:\;\r_1^{\otimes k_1}\otimes\cdots\otimes \r_{a}^{\otimes k_a}\prec\p\;\;,\;\;\forall\;\t\in\mbb{T}_{a,k}\Big\}\\
&\leq (k+\ell a)\log(n)-H(\v_1\otimes\cdots\v_a)
\ea
Continuing 
\be
H(\v_1\otimes\cdots\v_a)=\sum_{x=1}^{a}H(\v_x)=-a\delta\log\delta-\sum_{x=1}^{n}s_x'\log s_x'+\sum_{x=1}^{a}k_x\left(H(\r_z)-\epsilon\right)c_{x}'2^{-k_x\left(H(\r_z)-\epsilon\right)}
\ee
Now, since $k_x\eqdef kt_x+\ell$, we get that $\lim_{k\to\infty}\frac {k_x}k=t_x$ and 
\be
\lim_{k\to\infty}c_{x}'2^{-k_x\left(H(\r_z)-\epsilon\right)}=\left\lfloor(1-\delta)2^{k_x\left(H(\r_z)-\epsilon\right)}\right\rfloor2^{-k_x\left(H(\r_z)-\epsilon\right)}=1-\delta-g\delta_{0,t_x}
\ee
where
\be
g\eqdef\left\lfloor (1-\delta)2^{\ell\left(H(\r_z)-\epsilon\right)}\right\rfloor2^{-\ell\left(H(\r_z)-\epsilon\right)}-(1-\delta)
\ee
Therefore,
\ba
\lim_{k\to\infty}\frac1k\sum_{x=1}^{a}s_x\log s_x
&=\lim_{k\to\infty}\frac1k\sum_{x=1}^{a}\left(1-\delta-c_{x}' 2^{-k_x\left(H(\r_z)-\epsilon\right)}\right)\log\left(1-\delta-c_{x}' 2^{-k_x\left(H(\r_z)-\epsilon\right)}\right)\\
&=\lim_{k\to\infty}\frac1k\sum_{x=1}^{a}g\delta_{0,t_x}\log(g\delta_{0,t_x})=0
\ea
so that
\ba
\lim_{k\to\infty}\frac1kH(\v_1\otimes\cdots\v_a)&=\lim_{k\to\infty}\frac1k\sum_{x=1}^{a}k_x\left(H(\r_z)-\epsilon\right)c_{x}'2^{-k_x\left(H(\r_z)-\epsilon\right)}\\
&=\sum_{x=1}^{a}t_x\left(H(\r_z)-\epsilon\right)(1-\delta-g\delta_{0,t_x})\\
&=\left(H(\r_z)-\epsilon\right)(1-\delta)\;.
\ea
Combining this with~\eqref{11111} gives
\be
\bD^{\reg}(\mM\|\mN)\eqdef\lim_{k\to\infty}\frac1k\bD\left(\mM^{\otimes k}\big\|\mN^{\otimes k}\right)
\leq\log(n)- \left(H(\r_z)-\epsilon\right)(1-\delta)
\ee
Since the above equation holds for all $\epsilon,\delta>0$ we conclude that
\ba
\bD^{\reg}(\mM\|\mN)&\leq \log(n)- H(\r_z)\\
&=D(\r_z\|\u^{(n)})\\
&=\max_{x\in[a]}D(\r_x\|\u^{(n)})\\
&=\max_{x\in[a]}D(\m_x\|\n_x)\\
&=\uD(\mM\|\mN)\;.
\ea
Combining this with Theorem~\ref{thm2} and Theorem~\ref{outs} we get
\be
\uD^{\reg}(\mM\|\mN)\leq\D(\mM\|\mN)\leq\bD^{\reg}(\mM\|\mN)\leq \uD(\mM\|\mN)=\uD^{\reg}(\mM\|\mN)
\ee
where in the last equality we used the additivity of $\uD$ on classical channels. Therefore, all the inequalities above must be equalities so that $\D(\mM\|\mN)=\uD(\mM\|\mN)$.
This completes the proof for the case that the components of $\mN(|x\lr x|)$ are all positive rational numbers.
The proof for arbitrary classical channel $\mN$ then follows from the continuity of $\D(\mM\|\mN)$ in $\mN$. 
\end{proof}

\begin{center}\Large\bfseries {Proof of Theorem~\ref{outs}}\end{center}

\begin{theorem*}
Let $\mr$ be a function that maps any pair of quantum systems $A$ and $B$ to a subset $\mr(A\to B)\subset\cptp(A\to B)$, and let $\C$ be an $\mr$-divergence. Then, its maximal and minimal channel-extensions $\overline{C}$ and $\underline{C}$ have the following  properties:
\begin{enumerate} 
\item \textbf{Reduction.} For any $\mM,\mN\in\mr(A\to B)$
\be
\underline{\C}(\mM\|\mN)=\overline{\C}(\mM\|\mN)={\C}(\mM\|\mN)\;.
\ee
\item \textbf{Data Processing Inequality.} For any $\mM,\mN\in\cptp(A\to B)$ and any $\Theta\in\super(AB\to A'B')$
\be
\underline{\C}\big(\Theta[\mM]\big\|\Theta[\mN]\big)\leq\underline{\C}(\mM\|\mN)\quad\text{and}\quad\overline{\C}\big(\Theta[\mM]\big\|\Theta[\mN]\big)\leq\overline{\C}(\mM\|\mN)\;.
\ee
\item \textbf{Optimality.} Any quantum channel divergence $\sD$ that reduces to $\C$ on pairs of channels in $\mr(A\to B)$, must satisfy
\be\label{bounds91}
\underline{\C}(\mM\|\mN)\leq \sD(\mM\|\mN)\leq\overline{\C}(\mM\|\mN)\quad\forall\mM,\mN\in\cptp(A\to B)\;.
\ee
\item \textbf{Sub/Super Additivity.} If $\C$ is weakly additive under tensor products then $\underline{\C}$ is super-additive and $\overline{\C}$ is sub-additve. Explicitly, for any $\mM_1,\mM_2\in\cptp(A\to B)$ and any $\mN_1, \mN_2\in\cptp(A'\to B')$
\ba
&\underline{\C}\left(\mM_1\otimes\mM_2\big\|\mN_1\otimes\mN_2\right)\geq \underline{\C}(\mM_1\|\mN_1)+\underline{\C}(\mM_2\|\mN_2)\\
&\overline{\C}\left(\mM_1\otimes\mM_2\big\|\mN_1\otimes\mN_2\right)\leq \overline{\C}(\mM_1\|\mN_1)+\overline{\C}(\mM_2\|\mN_2)\;.
\ea
\item \textbf{Regularization.} If $\C$ is weakly additive under tensor products then any weakly additive quantum channel divergence $\sD$ that reduces to $\C$ on pairs of channels in $\mr(A\to B)$, must satisfy
\be\label{bounds101}
\underline{\C}^\reg(\mM\|\mN)\leq \sD(\mM\|\mN)\leq\overline{\C}^\reg(\mM\|\mN)\quad\forall\mM,\mN\in\cptp(A\to B)\;,
\ee
where
\be\label{xzx1}
\underline{\C}^\reg(\mM\|\mN)=\lim_{n\to\infty}\frac1n\underline{\C}\left(\mM^{\otimes n}\big\|\mN^{\otimes n}\right)\quad\text{and}\quad \overline{\C}^\reg(\mM\|\mN)=\lim_{n\to\infty}\frac1n\overline{\C}\left(\mM^{\otimes n}\big\|\mN^{\otimes n}\right)\;,
\ee
and $\underline{\C}^{\reg}$ and $\overline{\C}^{\reg}$ are themselves weakly additive normalized channel divergences.
\end{enumerate}
\end{theorem*}
\begin{proof}  In~\cite{GT2020b} the same theorem was proved for the quantum-state domain. The proof for the channel domain follows similar lines as we show below.
 
{\it Reduction.} This property follows from the definition and the fact that $\C$ satisfies the DPI when restricted to $\mr$ (see Definition~\ref{rdiv} of an $\mr$-divergence).

{\it Data Processing Inequality.} By definition,
\ba
\underline{\C}\big(\Theta[\mM]\big\|\Theta[\mN]\big)&=\sup_{A'',B''}\Big\{\C\big(\Upsilon\circ\Theta[\mM]\big\|\Upsilon\circ\Theta[\mN]\big)\;:\;\Upsilon\in\super(A'B'\to A''B''),\;\;\Upsilon\circ\Theta[\mM],\Upsilon\circ\Theta[\mN]\in\mr(A''\to B'')\Big\}\\
&\leq \sup_{A'',B''}\Big\{\C\big(\Gamma[\mM]\big\|\Gamma[\mN]\big)\;:\;\Gamma\in\super(AB\to A''B''),\;\;\Gamma[\mM],\Gamma[\mN]\in\mr(A''\to B'')\Big\}\\
&=\underline{\C}(\mM\|\mN)\;,
\ea
where the inequality follows from the replacement of $\Upsilon\circ\Theta$ with any $\Gamma\in\super(AB\to A''B'')$. Similarly,
let $\Theta\in\super(AB\to A''B'')$ and observe
\begin{align*}
&\overline{\C}(\mM\|\mN)\eqdef\inf_{A',B'}\Big\{\C(\mE\|\mF)\;:\;\mE,\mF\in\mr(A'\to B')\;\;,\;\;\exists\Upsilon\in\super(A'B'\to AB)\;\text{ s.t. }\mM=\Upsilon[\mE]\;\text{,}\;\mN=\Upsilon[\mF]\Big\}\\
&\geq\inf_{A',B'}\Big\{\C(\mE\|\mF)\;:\;\mE,\mF\in\mr(A'\to B')\;\;,\;\;\exists\Upsilon\in\super(A'B'\to AB)\;\text{ s.t. }\Theta[\mM]=\Theta\circ\Upsilon[\mE]\;\text{,}\;\Theta[\mN]=\Theta\circ\Upsilon[\mF]\Big\}\\
&\geq \inf_{A',B'}\Big\{\C(\mE\|\mF)\;:\;\mE,\mF\in\mr(A'\to B')\;\;,\;\;\exists\Gamma\in\super(A'B'\to A''B'')\;\text{ s.t. }\Theta[\mM]=\Gamma[\mE]\;\text{,}\;\Theta[\mN]=\Gamma[\mF]\Big\}\\
&=\overline{\C}\big(\Theta[\mM]\big\|\Theta[\mN]\big)\;,
\end{align*}
where the first inequality follows from the simple fact that if $\mM=\Upsilon[\mE]$ then also $\Theta[\mM]=\Theta\circ\Upsilon[\mE]$ (but the converse is not necessarily true), and the second inequality follows from the replacement of $\Theta\circ\Upsilon$ with any $\Gamma\in\super(A'B'\to A''B'')$.

{\it Optimality.} For any $\Theta\in\super(AB\to A'B')$ such that $\Theta[\mM],\Theta[\mN]\in\mr(A'\to B')$ we have
\be
\sD(\mM\|\mN)\geq \sD\big(\Theta[\mM]\big\|\Theta[\mN]\big)=\C\big(\Theta[\mM]\big\|\Theta[\mN]\big)\;.
\ee
Since the inequality above holds for all such $\Theta$ it holds also for the supremum over such $\Theta$. That is, 
$\sD(\mM\|\mN)\geq \underline{\C}(\mM\|\mN)$. Similarly, for any $\mE,\mF\in\mr(A'\to B')$ for which there exists $\Theta\in\super(A'B'\to AB)$ such that $\mM=\Theta[\mE]$ and $\mN=\Theta[\mF]$ we get
\be
\sD(\mM\|\mN)\leq\sD(\mE\|\mF)=\C(\mE\|\mF)\;.
\ee
Since the inequality above holds for all such $\mE,\mF$ it holds also for the infimum over such $\mE$ and $\mF$. That is, 
$\sD(\mM\|\mN)\leq \overline{\C}(\mM\|\mN)$. 

{\it Sub/Super Additivity.} The super-additivity of $\underline{\C}$ follows from the restriction of $\Theta$ in the definition of $\underline{\C}$ in~\eqref{min}  to have the form $\Theta_1\otimes\Theta_2$. Similarly, the sub-additivity of $\overline{\C}$ follows from the restriction of $\mE$, $\mF$, and $\Theta$ in~\eqref{max} to have  the tensor product form $\mE_1\otimes\mE_2$, $\mF_1\otimes\mF_2$, and $\Theta_1\otimes\Theta_2$, respectively.

{\it Regularization.} This property follows directly from the optimality property when replacing $\mM,\mN$ in~\eqref{bounds91} with $\mM^{\otimes n}$ and $\mN^{\otimes n}$, dividing everything by $n$, taking the limit $n\to\infty$, and using the additivity of $\sD$.
\end{proof}

\begin{center}\Large\bfseries {Proof of Theorem~\ref{maxu}}\end{center}

\begin{theorem*}
Let $\D$ be a channel divergence that reduces to $D_{\max}$ on classical probability distributions; i.e. for any classical system $X$ and $\p,\q\in\md(X)$, $\D(\p\|\q)=D_{\max}(\p\|\q)$. Then, for all $\mN,\mM\in\cptp(A\to B)$
\be
\D(\mN\|\mM)=D_{\max}(\mN\|\mM)\;.
\ee
\end{theorem*}

\begin{proof}
We need to show that for any $\mN,\mM\in\cptp(A\to B)$, the minimal and maximal extensions of $D_{\max}$ collapse into a single function given by
\be
\underbar{D}_{\max}(\mN\|\mM)=\bD_{\max}(\mN\|\mM)=D_{\max}(\mN\|\mM)\eqdef\log\min\Big\{t\in\mbb{R}\;:\;t\mM\geq\mN\Big\}\;.
\ee
Since $\uD_{\max}(\mN\|\mM)\leq\bD_{\max}(\mN\|\mM)$ it is sufficient to prove that $\uD_{\max}(\mN\|\mM)\geq D_{\max}(\mN\|\mM)$ and $\bD_{\max}(\mN\|\mM)\leq D_{\max}(\mN\|\mM)$. The latter follows directly from Property 6 of Theorem~\ref{properties}. 
For the first inequality, in the definition of $\uD_{\max}(\mN\|\mM)$, we take   $\Theta\in\super(AB\to X)$ to be the superchannel $$\Theta[\mN]\eqdef\Delta^{RB\to X}\circ\mN^{A\to B}\left(\phi_{+}^{RA}\right)$$ 
where $R\cong A$, $|X|=|RB|$, and $\Delta^{RB\to X}$ is the completely dephasing channel, where the classical basis of $X$ is chosen to consist of the eigenvectors of 
\be
s\mM^{A\to B}\left(\phi_{+}^{RA}\right)-\mN^{A\to B}\left(\phi_{+}^{RA}\right)\;,
\ee
where $s\eqdef2^{D_{\max}\left(\mM\left(\phi_{+}^{RA}\right)\big\|\mN\left(\phi_{+}^{RA}\right)\right)}$.
By definition
\ba
\underbar{D}_{\max}(\mN\|\mM)&\geq D_{\max}\Big(\Delta^{RB\to X}\circ\mN^{A\to B}\left(\phi_{+}^{RA}\right)\Big\|\Delta^{RB\to X}\circ\mM^{A\to B}\left(\phi_{+}^{RA}\right)\Big)\\
&=\log\min\Big\{t\;:\;\Delta^{RB\to X}\circ\left(t\mM^{A\to B}-\mN^{A\to B}\right)\left(\phi_{+}^{RA}\right)\geq 0\Big\}\\
&=\log\min\Big\{t\;:\;s\mM^{A\to B}-\mN^{A\to B}+\left(t-s\right)\Delta^{RB\to X}\circ\mM^{A\to B}\geq 0\Big\}\\
&=\log s
\ea
where the last equality follows from the fact that $(s\mM-\mN)(\phi_{+}^{RA})$ has one eigenvalue that is zero (since $s$ is the smallest number satisfying $s\mM\geq\mN$). Hence, if $t< s$ then $(t-s)\Delta\circ\mM(\phi_{+}^{RA})$, which commutes with $(s\mN-\mM)(\phi_{+}^{RA})$ will make the zero eigenvalue strictly negative.
Hence,
\be
\uD_{\max}(\mN\|\mM)\geq \log s=D_{\max}(\mN\|\mM)\;.
\ee
This completes the proof.
\end{proof}

\begin{center}\Large\bfseries {Proof of Theorem~\ref{lem1}}\end{center}

\begin{theorem*}
Let $\D$ be a jointly convex quantum divergence. Then, its maximal channel-extension $\bbd$ is also jointly convex.
\end{theorem*}

\begin{proof}
Let $\mN=\sum_xp_x\mN_x$, and $\mM=\sum_xp_x\mM_x$. Let $R=R'X$ and observe that
\begin{align*}
&\bbd_q(\mN\|\mM)= \inf_{\substack{|R'|\in\mbb{N}\;,\;\rho,\sigma\in\md(R'X)\\ \mN=\mE_\rho\;,\;\mM=\mE_\sigma\
\\\mE\in\cptp(R'XA\to B)}}\D_q(\rho\|\sigma)\leq \inf_{\substack{|R'|\in\mbb{N}\;,\;\rho_x,\sigma_x\in\md(R')\;\forall x\\ \sum_xp_x\mN_x=\sum_xp_x\mE_{\rho_x\otimes|x\lr x|}\\
\sum_xp_x\mM_x=\sum_xp_x\mE_{\sigma_x\otimes|x\lr x|}
\\\mE\in\cptp(R'XA\to B)}}\D_q\Big(\sum_xp_x\rho_x^{R'}\otimes|x\lr x|^X\Big\|\sum_xp_x\sigma_x^{R'}\otimes|x\lr x|^X\Big)\\
&=\inf_{\substack{|R'|\in\mbb{N}\;,\;\rho_x,\sigma_x\in\md(R')\;\forall x\\ \sum_xp_x\mN_x=\sum_xp_x\mE_{\rho_x\otimes|x\lr x|}\\
\sum_xp_x\mM_x=\sum_xp_x\mE_{\sigma_x\otimes|x\lr x|}
\\\mE\in\cptp(R'XA\to B)}}\sum_xp_x\D_q\big(\rho_x^{R'}\big\|\sigma_x^{R'}\big)
\leq \inf_{\substack{|R'|\in\mbb{N}\;,\;\rho_x,\sigma_x\in\md(R')\;\forall x\\ \mN_x=\mE_{\rho_x\otimes|x\lr x|}\;,\;
\mM_x=\mE_{\sigma_x\otimes|x\lr x|}
\\\mE\in\cptp(R'XA\to B)}}\sum_xp_x\D_q\big(\rho_x^{R'}\big\|\sigma_x^{R'}\big)
\end{align*}
Now, denoting by $\mE^{(x)}\in\cptp(R'A\to B) $ the channel
\be
\mE^{(x)}(\omega^{R'A})\eqdef \mE(\omega^{R'A}\otimes|x\lr x|^X)\;,
\ee
we continue
\be\nonumber
\bbd_q(\mN\|\mM)\leq  \inf_{\substack{|R'|\in\mbb{N}\;,\;\rho_x,\sigma_x\in\md(R')\;\forall x\\ \mN_x=\mE_{\rho_x}^{(x)}\;,\;
\mM_x=\mE_{\sigma_x}^{(x)}
\\\mE^{(x)}\in\cptp(R'A\to B)\;\;\forall x}}\sum_xp_x\D_q\big(\rho_x^{R'}\big\|\sigma_x^{R'}\big)
=\sum_xp_x\inf_{\substack{|R'|\in\mbb{N}\;,\;\rho_x,\sigma_x\in\md(R')\\ \mN_x=\mE_{\rho_x}^{(x)}\;,\;
\mM_x=\mE_{\sigma_x}^{(x)}
\\\mE^{(x)}\in\cptp(R'A\to B)}}\D_q\big(\rho_x^{R'}\big\|\sigma_x^{R'}\big)=\sum_xp_x\bbd_q(\mN_x\|\mM_x)
\ee
This completes the proof.
\end{proof}

\begin{center}\Large\bfseries {Proof of Theorem~\ref{lem17}}\end{center}

\begin{theorem*}
Let $\mM,\mN\in\cptp(A\to B)$ and for any $\epsilon\in[0,1)$, let $\xD=D_{\min}^{\epsilon}$ be the classical hypothesis testing divergence. 
Then,
\ba
\underline{\xD}(\mM\|\mN)=D_{\min}^{\epsilon}(\mM\|\mN)&\eqdef\sup_{\psi\in\md(RA)}D_{\min}^{\epsilon}(\mM^{A\to B}(\psi^{RA})\|\mN^{A\to B}(\psi^{RA}))\;.
\ea
\end{theorem*}

\begin{proof}
Observe that from the expression in~\eqref{minimalex}
\be
\uD_{\min}^{\epsilon}(\mM\|\mN)=\sup_{\psi,\mE}D_{\min}^\epsilon\Big(\mE_{BR\to X}\circ\mN_{A\to B}(\psi_{AR})\big\|\mE_{BR\to X}\circ\mM_{A\to B}(\psi_{AR})\Big)\;.
\ee
Let $T=\sum_{x}t_x|x\lr x|^X$ be a classical test operator with $0\leq t_x\leq 1$, and note that 
\ba
&\tr\left[T\mE_{BR\to X}\circ\mM_{A\to B}(\psi_{AR})\right]=\tr\left[\Gamma^{RB}\mM_{A\to B}(\psi_{AR})\right]\\
&\tr\left[T\mE_{BR\to X}\circ\mN_{A\to B}(\psi_{AR})\right]=\tr\left[\Gamma^{RB}\mN_{A\to B}(\psi_{AR})\right]
\ea
where $\Gamma^{RB}\eqdef\mE_{X\to RB}^{\dag}(T)$. Since $0\leq T\leq I$ also $0\leq\Gamma^{RB}\leq I^{RB}$, and for any $\Gamma^{RB}$ with this property there exists $0\leq T\leq I^X$ and $\mE\in\cptp(RB\to X)$ such that $\mE_{X\to RB}^{\dag}(T)=\Gamma^{RB}$. Hence, 
\ba
\uD_{\min}^{\epsilon}(\mM\|\mN)&=\sup_{\substack{\psi\in\md(RA)\;,\;0\leq\Gamma^{RB}\leq I^{RB}\\
\tr\left[\Gamma^{RB}\mN^{A\to B}(\psi^{RA})\right]\geq 1-\epsilon}}-\log\tr\left[\Gamma^{RB}\mN^{A\to B}(\psi^{RA})\right]\\
&=\sup_{\psi\in\md(RA)}D_{\min}^{\epsilon}\left(\mM^{A\to B}(\psi^{RA})\big\|\mN^{A\to B}(\psi^{RA})\right)\\
&=D_{\min}^{\epsilon}(\mM\|\mN)\;.
\ea
This completes the proof.
\end{proof}

\begin{center}\Large\bfseries {Proof of Theorem~\ref{eqreg}}\end{center}

\begin{theorem*}
Let $\mM,\mN\in\cptp(A\to B)$ be two quantum channels. Then,
\be
\uD^{\reg}(\mM\|\mN)= D^{\reg}(\mM\|\mN)\;.
\ee
That is, $D^{\reg}(\mM\|\mN)$ is the smallest weakly additive quantum-channel divergence that reduces to the KL relative entropy on classical states.
\end{theorem*}
\begin{proof}
Due to~\eqref{12345} it is left to show that $\uD^{\reg}(\mM\|\mN)\geq D^{\reg}(\mM\|\mN)$. We make use here of the inequality
\be
D_{\min}^{\epsilon}(\rho\|\sigma)\leq\frac{1}{1-\epsilon}(D(\rho\|\sigma)+h_2(\epsilon))
\ee
Combining this with Lemma~\ref{lem17} gives
\ba
D_{\min}^{\epsilon}(\mM\|\mN)=\uD_{\min}^{\epsilon}(\mM\|\mN)&=\sup_{\psi,\mE}D_{\min}^\epsilon\Big(\mE_{BR\to X}\circ\mN_{A\to B}(\psi_{AR})\big\|\mE_{BR\to X}\circ\mM_{A\to B}(\psi_{AR})\Big)\\
&\leq \frac{1}{1-\epsilon}\sup_{\psi,\mE}D\Big(\mE_{BR\to X}\circ\mN_{A\to B}(\psi_{AR})\big\|\mE_{BR\to X}\circ\mM_{A\to B}(\psi_{AR})\Big)+\frac{1}{1-\epsilon}h_2(\epsilon)\\
&=\frac{1}{1-\epsilon}\left(\uD(\mM\|\mN)+h_2(\epsilon)\right)\;.
\ea
Diving both sides by $n\in\mbb{N}$, replacing  the pair $(\mM,\mN)$ with $(\mM^{\otimes n},\mN^{\otimes n})$, and taking the limit $n\to \infty$ followed by $\epsilon\to 0^+$ gives
\be
\lim_{\epsilon \to 0^+}\liminf_{n\to\infty}\frac{1}{n}D_{\min}^{\epsilon}\left(\mM^{\otimes n}\|\mN^{\otimes n}\right)\leq \uD^{\reg}(\mM\|\mN)\;.
\ee
In~\cite{XinWilde2020} it was shown that the quantity on the LHS equals $D^{\reg}(\mM\|\mN)$. Therefore, this completes the proof.
\end{proof}

\begin{center}\Large\bfseries {Proof of Theorem~\ref{lem.3}}\end{center}

\begin{theorem*}
Let $\mV\in\cptp(A\to B)$ be an isometry channel defined via $\mV(\rho)=V\rho V^*$, for all $\rho\in\md(A)$, and with isometry matrix $V$ (i.e. $V^*V=I^A$). Then, for any $\mN\in\cptp(A\to B)$
\be
\overline{\xD}(\mV\|\mN)=D_{\max}(\mV\|\mN)=\log\tr\left[ J_\mN^{-1}J_\mV\right]\;.
\ee
\end{theorem*}

\begin{proof}
Since $\mV$ is an isometry, the condition $\mV=\sum_{x}p_{x}\mE_{x}$, can hold only if for all $x$ such that $p_x\neq 0$ we have $\mE_x=\mV$. W.l.o.g. let the $k$ first components of $\p$ be non-zero, while all the remaining components are zero. This implies that the second condition can be expressed as
\be
\mN=\sum_{x=1}^{k}q_x\mV+\sum_{x=k+1}^{n}q_{x}\mE_{x}\;.
\ee
Denote by $s\eqdef\sum_{x=1}^{k}q_x$, and observe that there exists such $\{\mE_{x}\}_{x=k+1}^{n}$ if and only if 
\be
\mN\geq s\mV
\ee
or in other words, iff $s^{-1}\geq 2^{D_{\max}(\mV\|\mN)}$. Consider the classical channel $\mC\in\cptp(X\to X)$ defined by
\be
\mC(|x\lr x|)=|1\lr 1|\;\;\forall x=1,...,k\quad\text{and}\quad\mC(|x\lr x|)=|2\lr 2|\;\;\forall x=k+1,...,n\;.
\ee
Therefore, we must have $\D(\p\|\q)\geq \D\big(\mC(\p)\big\|\mC(\q)\big)=\D\big(|1\lr 1|\big\|s|1\lr 1|+(1-s)|2\lr 2|\big)$.
This means that w.l.o.g. we can assume that $\p=|1\lr 1|$ and $\q$ is binary; i.e. $\q=s|1\lr 1|+(1-s)|2\lr 2|$ so that 
$\D(\p\|\q)=-\log(s)$ (cf.~\eqref{160}). But since we must have $s^{-1}\geq 2^{D_{\max}(\mV\|\mN)}$, the minimum value is achieved when $s^{-1}= 2^{D_{\max}(\mV\|\mN)}$. That is, $\D(\p\|\q)=D_{\max}(\mV\|\mN)$. This completes the proof.
\end{proof}

\end{document}